\documentclass[letterhpaper, 12pt, onecolumn, draftcls]{IEEEtran}
\usepackage{amsfonts}
\usepackage{amssymb}
\usepackage{stfloats}
\usepackage{cite}
\usepackage{graphicx}
\usepackage{psfrag}
\usepackage{subfigure}
\usepackage{amsmath}
\usepackage{array}
\usepackage{color}
\usepackage{mathrsfs}
\usepackage{bbding}
\usepackage{amsfonts}
\usepackage{amsmath}
\usepackage{graphicx}
\usepackage{subfigure}
\usepackage{latexsym}
\usepackage{amssymb}
\usepackage{amsmath}
\usepackage{bbding}
\usepackage{latexsym}

\begin{document}

\title{A Survey on Delay-Aware Resource Control for Wireless
Systems --- Large Deviation Theory, Stochastic Lyapunov Drift and Distributed
Stochastic Learning}

\newtheorem{Thm}{Theorem}
\newtheorem{Lem}{Lemma}
\newtheorem{Cor}{Corollary}
\newtheorem{Def}{Definition}
\newtheorem{Alg}{Algorithm}
\newtheorem{Prob}{Problem}
\newtheorem{Rem}{Remark}
\newtheorem{Proof}{Proof}
\newtheorem{Ass}{Assumption}
\newtheorem{Exam}{Example}

\author{Ying Cui\thanks{Ying Cui and Vincent K. N. Lau are with  Department of Electronic and Computer Engineering,
Hong Kong University of Science and Technology, Hong Kong. Rui Wang, Huang Huang
 and Shunqing Zhang
are with Huawei Technologies Co., Ltd., China. This work was supported by GREAT Project, Huawei Technologies Co., Ltd.} \, Vincent K. N. Lau  \, Rui Wang   \, Huang Huang
\, Shunqing Zhang}

\maketitle

\begin{abstract}
In this tutorial paper, a comprehensive survey is given on several
major systematic approaches in dealing with
delay-aware control problems, namely the {\em
equivalent rate constraint} approach, the {\em Lyapunov stability
drift} approach and the {\em approximate Markov
Decision Process (MDP)} approach using {\em
stochastic learning}. These approaches essentially embrace
most of the existing literature regarding delay-aware resource
control in wireless systems. They have their relative pros and cons
in terms of performance, complexity and implementation issues.
For each of the approaches, the problem setup, the
general solution and the design methodology are discussed.  Applications of these approaches to delay-aware resource allocation are
illustrated with examples in single-hop wireless networks.
Furthermore, recent results regarding delay-aware multi-hop routing
designs in general multi-hop networks  are elaborated. Finally, the
delay performance of the various approaches are compared through
simulations using an example of the uplink OFDMA
systems.
\end{abstract}


\section{Introduction}\label{sec_intro}
There is plenty of literature on
cross-layer resource optimization in wireless systems. For example,
there are papers on joint power and subcarrier
allocations to maximize the sum throughput for OFDMA
systems\cite{Svedman07,Pischella08}. There are also
papers on  joint power and precoder optimization to
boost the sum rate, weighted sum MMSE or SINR for MIMO wireless
systems\cite{Palomartxrsbeamforming:2003,PalomarlinearprocessingQos:2004}.
All these papers illustrate that significant throughput gain can be
obtained by joint optimization of radio resource across the Physical (PHY) and the Media Access
Control (MAC) layers. However, a typical assumption
in these papers is that the transmitter  has an infinite backlog and
the information flow is delay insensitive. As a result, these papers focus
only on optimizing the PHY layer performance metrics such as sum throughput, MMSE, SINR or proportional fairness, and the resulting control policy is adaptive to the channel state information (CSI) only.

\begin{figure}
\centering
\includegraphics[width = 14cm]{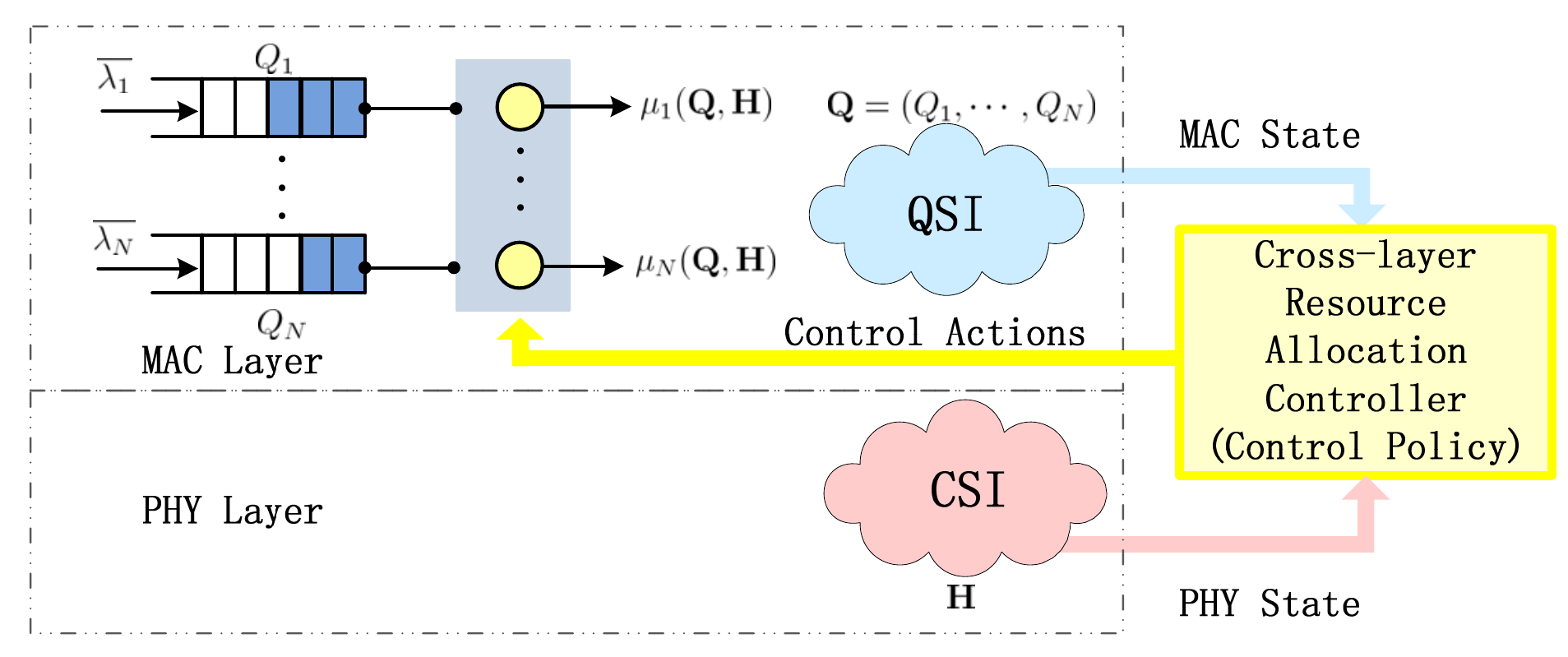}
\caption{Illustration of cross-layer resource allocation with
respect to both the MAC layer state (QSI) and PHY layer state
(CSI).} \label{Fig:Illustrate-Delay-Cross-Layer}
\end{figure}

In practice, it is very important to consider random bursty
arrivals and delay performance metrics in addition
to the conventional PHY layer performance
metrics in cross-layer optimization,
which may embrace the PHY, MAC and network layers.
A combined framework taking into account both queueing delay and
PHY layer performance is not trivial as it involves
both queueing theory (to model the queue dynamics) and   information
theory (to model the  PHY layer dynamics). The
system state involves both the CSI and the queue state information
(QSI) and the delay-optimal control policy should be adaptive to
both the CSI and the QSI of wireless systems as illustrated in
Fig. \ref{Fig:Illustrate-Delay-Cross-Layer}. This design approach
is fundamentally challenging for the following reasons.
First, there may not be closed-form
expressions relating the optimization objective
(such as the average delay) and the optimization variables (power,
precoder, etc). Second, it is not clear if the
optimization  problems are convex (in most cases,
they are not convex). Third, there is the {\em
curse of dimensionality} due to the exponential growth of the
cardinality of the system state space as well as the
large dimension of the control
action space involved (i.e., set of actions). For example,
consider a queueing network with $N$ queues, each with finite buffer
size $N_Q$. The size of the system state space is $\mathcal O
(N_Q^{N})$, which is unmanageable even for small number of users $N$
and buffer length $N_Q$.

There are various approaches to deal with delay-aware resource
control in wireless
networks\cite{MengChiang08:NUMtutorial,LinShroffSrikant06:crosslayertutorial}.
One approach converts average delay
constraints into equivalent average rate
constraints using the large deviation theory and
solves the optimization problem using a purely
information theoretical formulation based on the rate
constraints
\cite{Wu03,Hui07,Tang0705,Tang0708,Tang0806,Melissa08}. While this
approach allows potentially simple solutions, the
resulting control  policies are only
functions of the CSI and such controls are good only for
the large delay regime where the probability of
empty buffers is small. In general,  optimal
control policies should be functions of both the CSI
and QSI. In addition, due to the complex coupling
among  queues in multi-hop wireless networks, it is difficult to
express the average delay in terms of all the control actions.
Therefore, it is  not easy to generalize this approach to joint
resource allocation and routing in multi-hop wireless networks.

A second approach to deal with
delay-aware resource control utilizes the notion of {\em Lyapunov
stability} and establishes {\em throughput-optimal} control
policies (in the stability sense).
The throughput-optimal policies ensure the stability of the queueing network if stability can be  indeed achieved
under any policy. Three classes of policies that
are known to be throughput-optimal include the Max
Weight rule \cite{Georgiadis-Neely-Tassiulas:2006}, the Exponential
(EXP) rule \cite{EXPruleShakkottai-Stolyar:2002} and the Log
rule\cite{logrule:2009}. Among the three classes,
the throughput-optimal property of the Max Weight type algorithms
\cite{StolyarMLWDF:2004} and the Log rule \cite{logrule:2009} are both proved by the theory of {\em Lyapunov
drift}, whereas the EXP rule is
proved to be throughput-optimal by the {\em fluid limit} technique
along with a {\em separation of time scales} argument in
\cite{EXPruleShakkottai-Stolyar:2002}. Specifically, the general Max
Weight type algorithms are proved
to minimize the Lyapunov drift, and hence, are
throughput-optimal. Many dynamic control algorithms
belong to this type, which include
optimizing the allocation of computer resources
\cite{BhattacharyaAdaptiveLexicographicOpt:1993}, and stabilizing
packet switch systems \cite{Keslassy-McKeown:2001,Kumar-Meyn:1995,
Leonardi:INFOCOM:2001, McKeown-M-A-W:1999} and  satellite and
wireless systems
\cite{Tassiulas-Ephremides:1993-2,Andrew-Kumaran-Stolyar:2001,Neely-Modiano-Rohrs:2003}.
The {\em Lyapunov drift} theory (which only focuses
on controlling a queueing network to achieve  stability)
is extended to the {\em Lyapunov optimization}
theory (which enables stability and performance optimization to be
treated
simultaneously)\cite{Neely:PhD:2003,Neely-Modiano-Rohrs:2005,Neely:2006,Georgiadis-Neely-Tassiulas:2006}.
For example, utilizing the Lyapunov optimization theory, the
{\em Energy-Efficient Control Algorithm} (EECA)  proposed in
\cite{Neely:2006} stabilizes the system and consumes an average power
that is arbitrarily close to the minimum power solution
with a corresponding tradeoff in network delay. In
transport layer flow control and network fairness optimization,
the Cross Layer Control (CLC) algorithm was
designed in \cite{Neely:PhD:2003} to achieve a fair
throughput point which is arbitrarily close to optimal
with a corresponding tradeoff in network delay,
when the exogenous arrival rates are outside of the network
stability region. In \cite{Berry-Gallager:2002} and
\cite{Neely:2007}, the authors consider the
asymptotic single-user and multi-user power-delay tradeoff
in the large delay regime and
obtain insights into the structure of the optimal
control policy in the large delay regime. Although the derived policy (e.g.,
dynamic backpresssure algorithm) by the Lyapunov drift theory and the Lyapunov
optimization theory may not have good delay performance in moderate and
light traffic loading regimes, it allows potentially simple
solutions with throughput optimality in multi-hop
wireless networks. However, throughput optimality is a weak form of
delay performance and it is also of great interest to study
scheduling policies that minimize
average delay of queueing networks.

A more systematic approach in dealing with delay-optimal resource
control in  general delay regime is  the {\em Markov
Decision Process} (MDP) approach. In some special cases, it may be
possible to obtain simple delay-optimal solutions. For example, in
\cite{Yeh:PhD:2001,Yeh-Cohen:2003}, the authors utilize {\em
Stochastic Majorization} to show that the longest queue highest
possible rate (LQHPR) policy is delay-optimal for multiaccess
systems with homogeneous users. However, in general,  the
delay-optimal control belongs to the infinite horizon average cost
MDP, and it is well known that there is no simple solution
associated with such MDP. Brute force value iterations or policy
iterations \cite{Bertsekas:2007,Cao:2008} could not lead to any
viable solutions due to the curse of dimensionality. In addition to
the above challenges, the problem is further
complicated under distributed implementation requirements.  For
instance, the delay-optimal control actions should be adaptive to
both the global system CSI and QSI. However, these CSI and QSI
observations are usually measured locally at some nodes of the
network and hence, centralized solutions require huge signaling
overhead to deliver all these local CSI and QSI to the centralized controller.
It is very desirable to have distributed solutions
where the control actions are computed locally based on the local CSI
and  QSI measurements.

A systematic understanding of  delay-aware control in wireless
communications is the key to truly embracing both
the PHY layer and the MAC layer in
cross-layer designs.  In this paper, we   give a
comprehensive survey on the major systematic approaches in dealing
with delay-aware control problems, namely the {\em equivalent
rate constraint} approach, the {\em Lynapnov stability drift}
approach and the {\em approximate MDP} approach using {\em
stochastic learning}. These approaches essentially embrace most of
the existing literature regarding delay-aware resource control in
wireless systems. They have their relative pros and cons in terms of
performance, complexity and implementation issues. For each of the approaches, we discuss the problem
setup, the general solution, the design methodology and the
limitations of delay-aware resource allocations  with simple examples in single-hop
wireless networks. We  also discuss recent
advances in  delay-aware routing designs in multi-hop wireless
networks.

The paper is organized as follows. In Section~\ref{sec_sys_mod}, we
  elaborate on the basic concepts of cross-layer resource
allocation, which consists of the system model, the
source model, the control policies, the queue
dynamics and the general resource control problem formulation. In
Section~\ref{sec_Rate_Constraint_formulation}, we elaborate
on the theory and the framework of the first
approach (equivalent rate constraint). In
Section~\ref{sec_Lynapnov}, we elaborate on the
theory and the framework of the second approach (Lynapnov stability
drift). In Section~\ref{sec_MDP}, we   elaborate on
the theory and the framework of the third approach (MDP) and
illustrate how the approximate MDP and stochastic learning could
help to obtain low complexity and distributed delay-aware control
solutions. In Section \ref{sec_routing}, we
discuss the delay-aware routing designs in multi-hop wireless
networks. In Section~\ref{sec_comparison}, we compare the
performance of the aforementioned approaches in a common application
topology, namely the {\em uplink OFDMA systems with multiple users}.
Finally, we conclude with a brief summary of the
results in Section~\ref{sec_summary}.

\section{System Model and General Cross-Layer Optimization Framework}
\label{sec_sys_mod}

In this section, we  elaborate on the system model,
the queue model, the framework of resource control for general
wireless networks. We also use the uplink OFDMA
systems as an example in the elaboration to make the description easy
to understand.

\subsection{System Model}
\label{sec_sys_mod_1} In this paper, we study delay-aware resource
control in a general multi-hop wireless network with a set of $N$
nodes ${\cal N} = \{1,2,...,N\}$ and a set of $L$ transmission links
${\cal L}=\{1,2,...,L\}$ as illustrated in Fig.
\ref{Fig:Illustrate-System-Graph}. Each link in set $\cal L$
denotes a communication channel for direct transmission from  node
$s\in \mathcal{N}$ to node $d \in \mathcal{N}$, and
is labeled by the ordered
pair\footnote{Note that $(s,d)$ and $(d,s)$
denote two different transmission links: the former is the link from
the $s$-th node to the $d$-th node, whereas the latter is the link
from the $d$-th node to the $s$-th node.} $(s,d)$. We denote $s(l)$
and $d(l)$ as the transmit node and the receive node of the $l$-th
link,
respectively.
\begin{figure}
\centering
\includegraphics[height=7cm, width=14cm]{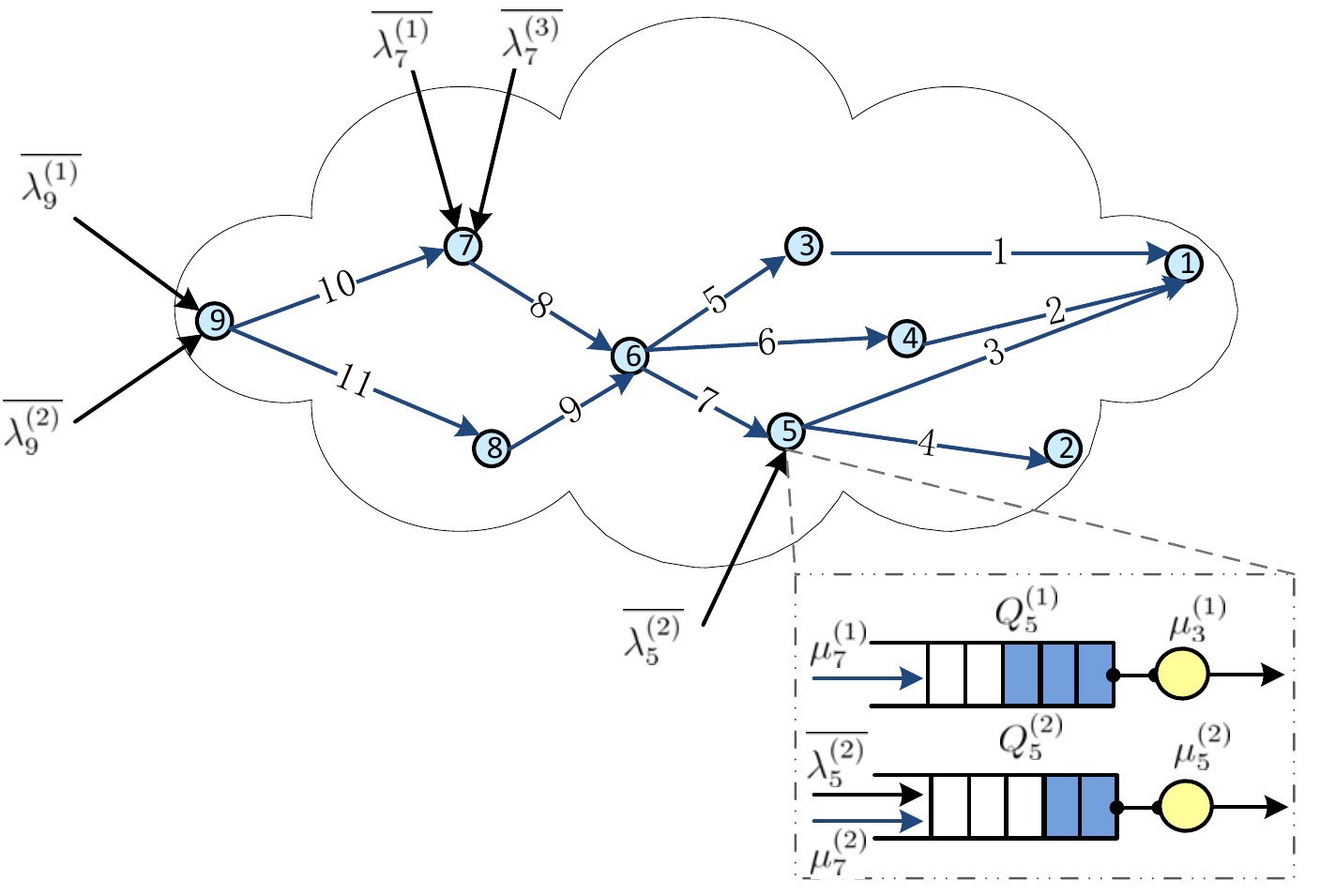}
\caption {Illustrative diagram of a multi-hop wireless network with
${\cal N} = \{1,2,...,9\}$, ${\cal L}=\{1,2,...,1\}$ and ${\cal
C}=\{1,2,3\}$.} \label{Fig:Illustrate-System-Graph}
\end{figure}

The network is assumed to work in slotted time with slot boundaries
that occur at time instances
$t \in \{1,2,...\}$. We use slot $t$ to denote the time interval
$[t,t+1)$. Denote
$\mathbf{H}(t)=[\mathbf{H}_1(t),\mathbf{H}_2(t),...,\mathbf{H}_L(t)]\in
\mathcal H$ as the CSI of all $L$  links in set $\mathcal{L}$ in  slot
$t$, where $\mathcal H$ denotes the system CSI state space. We have
the following assumption on the channel fading.

\begin{Ass}[Assumption on the Channel Fading]
Each element in $\mathbf{H}$ takes value from the
discrete state space $\mathcal H$ and the system CSI $\mathbf{H}(t)$
is a Markov process, i.e., $$\Pr
\big[\mathbf{H}(t)|\mathbf{H}(t-1),\mathbf{H}(t-2),...,\mathbf{H}(0)
\big]= \Pr \big[\mathbf{H}(t)|\mathbf{H}(t-1)\big].$$
~\hfill\QEDclosed\label{ass:channel}
\end{Ass}
The general network model described above
encompasses  a wide range of
practical network topologies.

\subsection{Source Model}
All data that enters the network is associated with
a particular commodity\footnote{The commodity index $c$ can be
interpreted as the {\em data flow index} in the network.} $c\in
\mathcal C$, which minimally defines the destination of the data,
but might also specify other information, such as the source node of
the data  or its priority service
class\cite{Georgiadis-Neely-Tassiulas:2006}. $\mathcal
C=\{1,2,\cdots, C\}$ represents the set of $C$ commodities in the
network.  Let $\lambda_n^{(c)}(t)$ denote the amount of new
commodity $c$ data (in number of bits) that
exogenously arrives to
 node $n$ at the end of slot $t$. We make the following assumption on
the arrival process.
\begin{Ass}[Assumption on Arrival Process]
The packet arrival process $\lambda_n^{(c)}(t)\in [0,\lambda_{n,\max}^{(c)}]$ is i.i.d. over
scheduling slots following general distribution with average arrival
rate $\mathbf
E[\lambda_n^{(c)}(t)]=\overline{\lambda_n^{(c)}}$. ~
\hfill\QEDclosed \label{ass:arrival}
\end{Ass}

Each node $n$ maintains a set of queues for storing
data according to its commodity. Let $Q_n^{(c)}(t)$ denote the queue
length (in number of bits) of commodity $c$ stored at node $n$. Note
that we let $Q_n^{(c)}(t)=0$ for all $t$ if node $n$ is the
destination of commodity $c$. Let $\mu_l^{(c)}(t)$ denote the rate
offered to commodity $c$ over  link $l$ during  slot $t$. Therefore,
the system queue dynamics is given by
\cite{Georgiadis-Neely-Tassiulas:2006}
\begin{equation}
Q_n^{(c)}(t+1) \leq   \max\left\{Q_n^{(c)}(t) -
\sum_{l\in \{l:s(l)=n\}}\mu_l^{(c)}(t),0\right\} +
\lambda_n^{(c)}(t)+\sum_{l\in \{l:d(l)=n\}}\mu_l^{(c)}(t), \quad \ n
\in \mathcal N.\label{eqn:sec2:queue}
\end{equation}
The above expression is an inequality rather than an
equality because the actual amount of commodity $c$ data arriving to
node $n$ during slot $t$ may be less than $\sum_{l\in
\{l:d(l)=n\}}\mu_l^{(c)}(t)$ if the neighboring nodes have little or
no commodity $c$ data to transmit. For notational
convenience, we define the QSI as $\mathbf{Q}(t) =
\big[Q_n^{(c)}(t)\big]\in \mathcal Q$, where $\mathcal Q$
denotes the system QSI state space.

\subsection{Control Policy and Resource Control Framework}

Let $\chi(t) = \{\mathbf{H}(t), \mathbf{Q}(t)\} \in \mathcal X$ be
the system state which can be estimated by the resource controller
at the $t$-th slot, where $\mathcal X= \mathcal H \times \mathcal Q$
is the full system state space. In practice, different control
policies may be adaptive to partial or full system
states. For example, a {\em CSI-only} control policy has control
actions that are adaptive to the partial system
state CSI only. A {\em QSI-only} control policy has control actions that are
adaptive to the partial system state QSI only. A {\em cross-layer}
control policy has control actions that are adaptive to the full
system state, i.e., the CSI and the QSI. We define $\Omega: \mathcal X \to
\mathcal A $ to be the control policy, which is a mapping from the
full system state space $\mathcal X$ to the action space $\mathcal
A$. The control policy may include the resource allocation policy
(e.g., power allocation policy, subcarrier allocation policy, precoder design policy,
etc) and the routing policy.

Under control policy $\Omega$, the average queue
length of commodity $c$ stored at node $n$ is given by
$$\overline{Q_n^{(c)}}=\limsup\limits_{T\rightarrow +\infty}\frac{1}{T}\sum\limits_{t=1}^{T}
\mathbf{E}^{\Omega} [Q_n^{(c)}(t)], \quad \forall n \in \mathcal{N},
c\in \mathcal C,$$ where $\mathbf{E}^{\Omega} [ \cdot]$ means the
expectation operation taken  w.r.t. the  measure induced
by the given policy $\Omega$. We also introduce the
average drop rate as a performance metric in our general system
model to incorporate delay-aware resource control
in queueing networks with finite buffer
size (c.f.,
\cite{Bertsekas:2007}), where data dropping is
necessary when a buffer overflows. For  queueing networks with finite buffer size $N_Q$, the average
drop rate of commodity $c$ stored at node $n$ is defined as
\begin{align}
\overline{d_n^{(c)}}=\limsup\limits_{T\rightarrow
+\infty}\frac{1}{T}\sum\limits_{t=1}^{T} \mathbf{E}^{\Omega}
\big[\mathbf{I}[Q_n^{(c)}(t)=N_Q]  \big], \quad \forall n \in
\mathcal{N}, c\in \mathcal C. \nonumber
\end{align}
Taking the effect of data dropping into consideration, we refer
to the average delay as the
average time that a piece of data stays in the
network before reaching the destination (averaged over the data that
are not dropped\footnote{For example, suppose
100 packets enter a single-hop network, among
which, 10 packets are dropped and the other 90 packets are
successfully delivered to their destinations . Furthermore, the
total time taken by the 90 packets to reach their destinations is
90. The average delay is given by $90 / 90 = 1$ and the average drop
rate is given by $10/100=0.1$.}). This is because the penalty of
data dropping is accounted for separately
in the average drop rate. The following lemma
extends Little's Law to the case with data dropping.
\begin{Lem} [Little's Law with Data Dropping]
The average delay of all the commodities and commodity $c$ in the
network are given by
\begin{align}\overline{D}\leq& \frac{\sum_{n \in \mathcal{N}} \sum_{c \in
\mathcal{C}} \overline{Q_n^{(c)}}}{\sum_{n \in \mathcal{N}} \sum_{c
\in \mathcal{C}} (1-\overline{d_n^{(c)}})\overline{\lambda_n^{(c)}}},
\label{eqn:average-delay-un-drop}\\
\overline{D^{(c)}}\leq&\frac{\sum_{n \in \mathcal{N}}
\overline{Q_n^{(c)}}}{\sum_{n \in \mathcal{N}}
(1-\overline{d_n^{(c)}})\overline{\lambda_n^{(c)}}},
\label{eqn:average-delay-un-drop-c}
\end{align}
where the above two inequalities are asymptotically tight for
general multi-hop networks as $\overline{d_n^{(c)}}\to 0$ for all $
n\in \mathcal N$ and  $c\in \mathcal C$. In addition, in single-hop queueing networks, the
inequalities in \eqref{eqn:average-delay-un-drop} and
\eqref{eqn:average-delay-un-drop-c} are tight for any
$\overline{d_n^{(c)}}$ and $\lambda_{n,\max}^{(c)}=1$, and  the average delay of commodity $c$ at
node $n$ is given by
\begin{align}\overline{D_n^{(c)}}=\frac{
\overline{Q_n^{(c)}}}{(1-\overline{d_n^{(c)}})\overline{\lambda_n^{(c)}}}.
\nonumber
\end{align}
\label{Lem:littleLaw-drop} ~\hfill\QEDclosed
\end{Lem}
\begin{proof}
The proof can be easily extended from the  standard Little's law
\cite{Kleinrock75} by considering the data that are not
dropped\footnote{Since all the data (including  the
data that are ultimately
 dropped) contributes to the queue length
process (before being dropped), we have inequalities in
\eqref{eqn:average-delay-un-drop} and
\eqref{eqn:average-delay-un-drop-c}.}. We omit the details due to
page limit.
\end{proof}
\begin{Rem} [Interpretation of Lemma \ref{Lem:littleLaw-drop}]
Lemma \ref{Lem:littleLaw-drop} establishes the
relationship among the average
delay, the average  queue length and the average drop
rate in
general networks.
Given the average drop rate, the average delay bound for general
multi-hop networks or the average delay for single-hop networks is
proportional to the average queue length. Thus, the average queue
length (and the average drop rate if data dropping happens) is
commonly used in the existing literature
\cite{Georgiadis-Neely-Tassiulas:2006} as the delay
performance measure. ~\hfill\QEDclosed
\end{Rem}

Moreover, under control policy $\Omega$, the average throughput of
link $l \in \mathcal{L}$ and the average power consumption of  node
$n\in \cal N$ are given by
\begin{align}
\overline{T_l} =& \limsup\limits_{T\rightarrow
+\infty}\frac{1}{T}\sum\limits_{t=1}^{T} \mathbf{E}^{\Omega} \left[
\sum_{c\in \mathcal C} \mu_l^{(c)}(t) \right], \quad \forall l \in
\mathcal{L} \nonumber\\
 \overline{P_n}=&\limsup\limits_{T\rightarrow
+\infty}\frac{1}{T}\sum\limits_{t=1}^{T} \mathbf{E}^{\Omega} \left[
\sum_{l\in \{l:s(l)=n\}}\sum_{c\in \mathcal C} p_l^{(c)}(t) \right],
\quad \forall n \in \mathcal{N} \nonumber
\end{align}
respectively, where $ p_l^{(c)}(t)$ denotes the power allocated to
commodity $c$ over link $l$ at slot $t$.

Therefore,  delay-aware resource control problems for wireless
networks can be divided into the following three
categories:
\begin{itemize}
\item {\bf Category I:} Maximize the average weighted sum system throughput (or average
arrival rate) subject to average delay
constraints, average power
constraints and average drop rate
constraints. Thus, the delay-aware
resource control problem can be expressed as
\begin{eqnarray}
&\max & \sum_{l\in \mathcal{L}} w_{l} \overline{T_{l}} \label{eqn:catI}\\
&s.t. & \overline{Q_n^{(c)}} \leq Q_n^{(c)}, \quad  \forall n \in
\mathcal{N}, c\in \mathcal C \nonumber\\
&& \overline{P_n} \leq P_n, \quad \forall n \in \mathcal{N}
\nonumber\\
&& \overline{d_n^{(c)}} \leq d_n^{(c)}, \quad  \forall n \in
\mathcal{N}, c\in \mathcal C, \nonumber
\end{eqnarray}
where $w_{l}$ is the weight for the $l$-th link, $Q_n^{(c)}$, $P_n$
and $d_n^{(c)}$ are the average delay constraint,
the average power constraint and the average drop
rate constraint for commodity $c$ at node $n$, respectively.

\item {\bf Category II:} Minimize the average weighted sum delay subject to average power constraints and average  drop
rate constraints for given arrival rates at all
sources. Thus, the delay-aware resource control problem can
be expressed as
\begin{eqnarray}
&\min & \sum_{n \in \mathcal{N}} \sum_{c \in \mathcal{C}}  w_n^{(c)} \overline{Q_n^{(c)}} \label{eqn:catII}\\
&s.t.& \overline{P_n} \leq P_n, \quad \forall n \in \mathcal{N}
\nonumber\\
&&\overline{d_n^{(c)}} \leq d_n^{(c)}, \quad  \forall n \in
\mathcal{N}, c\in \mathcal C, \nonumber
\end{eqnarray}
where $w_n^{(c)}$ is the weight for commodity $c$ at node $n$.

\item {\bf Category III:} Minimize the average weighed sum power consumption subject
to average delay constraints and average  drop rate
constraints for given arrival rates at all sources.
Thus, the delay-aware resource control problem can be
expressed as
\begin{eqnarray}
&\min & \sum_{n \in \mathcal{N}} w_{n} \overline{P_n} \label{eqn:catIII}\\
&s.t. & \overline{Q_n^{(c)}} \leq Q_n^{(c)}, \quad  \forall n \in
\mathcal{N}, c\in \mathcal C \nonumber\\
&& \overline{d_n^{(c)}} \leq d_n^{(c)}, \quad  \forall n \in
\mathcal{N}, c\in \mathcal C, \nonumber
\end{eqnarray}
where $w_n$ is the weight for node $n$.
\end{itemize}

\begin{Rem}[Unified Optimization Framework]
Note that the Lagrangian function of all the above optimization
problems can be written in a unified form:
\begin{eqnarray} L &=& \sum_{l \in \mathcal{L}}
\xi_{l} \overline {T_l}+ \sum_{n \in \mathcal{N}} \sum_{c \in
\mathcal{C}}  \left( \nu_n^{(c)} \overline{Q_n^{(c)}}+ \eta_n^{(c)}
\overline{d_n^{(c)}}\right) + \sum_{n \in \mathcal{N}}\gamma_{n}
\overline{P_n}, \nonumber
\end{eqnarray}
where $\xi_{l}$, $ \nu_n^{(c)}$, $\eta_n^{(c)} $ and $\gamma_{n}$
can be Lagrange Multipliers associated with the constraints or
weights in the objective function. Hence, these
problems can be solved by a
common optimization framework.~\hfill\QEDclosed
\end{Rem}

\subsection{Uplink OFDMA Systems}

\begin{figure}
\centering
\includegraphics[width = 12cm]{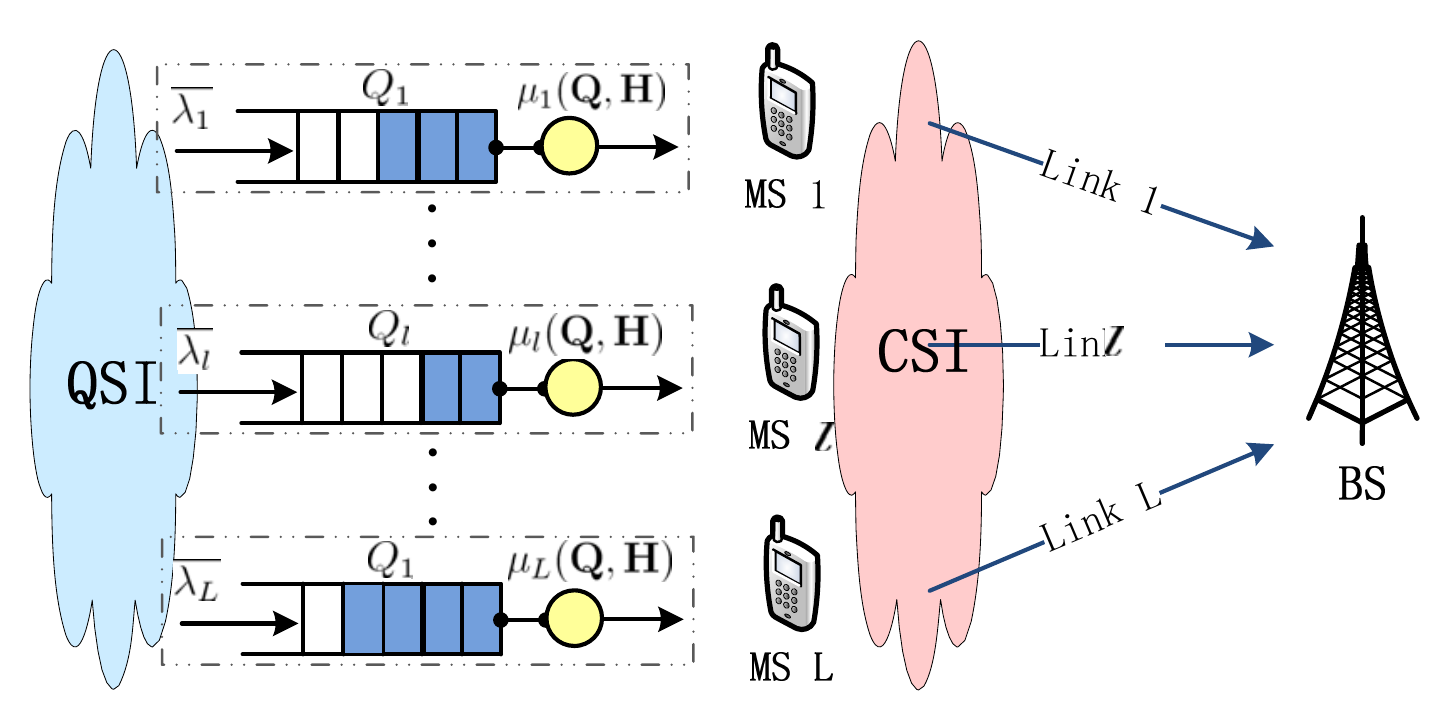}
\caption {Block diagram of uplink OFDMA systems.} \label{Fig:OFDMA}
\end{figure}

In this part, we  illustrate the general network
model in Section \ref{sec_sys_mod_1} with a
simple example of  one-hop uplink OFDMA systems.
This example topology will also be used as illustration to the
delay-aware resource control in later sections.

In the uplink OFDMA system example illustrated in
Fig. \ref{Fig:OFDMA}, we assume the set of $N$ nodes $\mathcal N$
are mobile stations (MSs) that communicate
with one base station (BS). Each MS
and the BS is equipped with a single antenna. Therefore, the set of
links $\cal L$ corresponds to the set of all uplink channels from
the $N$ MSs to the BS (with $L=N$). Furthermore, there are $N$ data
flows (with $C=N$), and for notation simplicity, we
use $l\in \{1,2,\cdots,N(=L)\}$ to denote the link index,
the node index as well as the commodity index. We
consider communications over a wideband frequency selective fading
channel, and the whole spectrum is divided into $N_F$ orthogonal
flat fading frequency bands (subcarriers). Let $H_{l,m}(t)$ denote the CSI of the
$l$-th uplink on the $m$-th ($m
\in\{1,2,\cdots,N_F\}$) subcarrier and let
$\mathbf{H}_l(t)=[H_{l,1}(t),H_{l,2}(t),...,H_{l,N_F}(t)]$ denote
the aggregate CSI over all subcarriers of the
$l$-th link. The system CSI
$\mathbf{H}(t)=[\mathbf{H}_1(t),\mathbf{H}_2(t),...,\mathbf{H}_N(t)]$
is a Markov process satisfying  Assumption \ref{ass:channel}. In
addition, we assume $\{H_{l,m}(t)\}$ are i.i.d. w.r.t. $l
\in \{1,2,\cdots,L\}$ and $m \in\{1,2,\cdots,N_F\}$.
Let $Q_l(t)$ and $s_{l,m}(t)\in
\{0,1\}$ denote the queue length and the subcarrier allocation for
the $l$-th link on the $m$-th subcarrier at slot $t$, respectively.
The received signal from the $l$-th user on the $m$-th subcarrier of
the BS at slot $t$ is given by
\begin{equation}
Y_{l,m}(t) =
s_{l,m}(t)\left(H_{l,m}(t)
X_{l,m}(t) + Z_{l,m}(t)\right), \ l
\in \mathcal L, \ m \in\{1,2,\cdots,N_F\}, \nonumber
\end{equation}
where $X_{l,m}(t)$ is the transmit symbol and
$Z_{l,m}(t)\sim \mathcal{CN} (0,1)$ is the channel
noise of the $l$-th link on the $m$-th subcarrier at
slot $t$. Hence, the data rate of the $l$-th link on the $m$-th
subcarrier at slot $t$ is given by
\begin{equation}
R_{l,m} (t)=s_{l,m}(t)
\log_2\big(1 + p_{l,m}(t)
|H_{l,m}(t)|^2\big), \ l \in \mathcal L, \ m
\in\{1,2,\cdots,N_F\},
 \nonumber
\end{equation}
where $p_{l,m}(t)$ is the transmit power over the
$l$-th link on the $m$-th subcarrier at slot $t$.
The sum rate of the $l$-th link at slot $t$ is
given by $\mu_l(t)=\sum_{m=1}^{N_F}
R_{l,m}(t)$.

In this uplink OFDMA system example, the control policy for
the $l$-th link is given by $\Omega_l=(\Omega_{l,p}
, \Omega_{l,s})$, where the power allocation policy $\Omega_{l,p}$
and the subcarrier allocation
policy\footnote{Please note that when $N_F=1$,
i.e., there is only one carrier, the subcarrier
allocation is reduced to link selection. Therefore,
the subcarrier allocation policy considered in the following problem
formulation covers most of the cases in resource allocations for
single-hop wireless networks.} $\Omega_{l,s}$  are defined as
follows.
\begin{Def}[Power Allocation Policy]
The power allocation policy of the $l$-th link is a mapping
$\mathcal X \to \mathcal P_l$ from the system state to the power
allocation action, which is given by
\begin{equation}
\Omega_{l,p} (\chi)= \Big \{p_{l,m}\geq0: m \in
\{1,2,\cdots,N_F\}\Big\}\in \mathcal P_l, \ \forall l \in \mathcal L,
\end{equation}
where $p_{l,m}$ is the transmit power on the $m$-th subcarrier
of the $l$-th link.~\hfill\QEDclosed
\end{Def}

\begin{Def}[Subcarrier Allocation Policy]
The subcarrier allocation policy of the $l$-th link is a mapping
$\mathcal X \to \mathcal S_l$ from the system state to the
subcarrier allocation action, which is given by
\begin{equation}
\Omega_{l,s}(\chi) = \Big \{s_{l,m}\in \{0,1\}: m \in
\{1,2,\cdots,N_F\}\Big\}\in \mathcal S_l, \ \forall l \in \mathcal L,
\end{equation}
where $s_{l,m}=1$ means that the $m$-th subcarrier
is used by the $l$-th link for data transmission,
and $s_{l,m}=0$ otherwise.~\hfill\QEDclosed
\end{Def}

\section{Equivalent Rate Constraint Approach}
\label{sec_Rate_Constraint_formulation}

The first attempt in the literature to deal with
the complicated delay control problem is to consider an equivalent
problem in the PHY layer domain only,
i.e., converting average delay constraints into
average rate constraints using the large deviation
theory \cite{Wu03,Hui07,Tang0708,Tang0806,Melissa08}. This approach
can be traced back to the early 90's,
when the statistical quality of
service (QoS) requirements have been extensively
studied in the context of {\em effective bandwidth theory}
\cite{Chang94,Chang95,Kim00,Kesidis93}, which asymptotically models
the statistical behavior of a source traffic process in the wired
networks (e.g., asynchronous transfer mode (ATM)
and Internet protocol (IP) networks).

For notational simplicity, we consider a single
queue in the following introduction of the known results on
the large deviation theory. Let $A(t)$ represent
the amount of source data (in number of bits) over
the time interval $[0,t)$. Assume that the
Gartner-Ellis limit of $A(t)$, expressed as
$\Lambda_B(\theta)=\lim_{t\to\infty}\frac{\log \mathbf E[e^{\theta
A(t)}]}{t}$ exists for all $\theta\geq0$. Then, the {\em effective
bandwidth} function of $A(t)$ is defined as
\begin{align}
E_B(\theta)=\frac{\Lambda_B(\theta)}{\theta}=\lim_{t\to\infty}\frac{1}{\theta
t}\log \mathbf E\left[e^{\theta A(t)}\right].\label{eqn:E-B}
\end{align}
Consider a queue with infinite buffer size served by a channel with
constant service rate $R$. By the large deviation
theory\cite{Shwartz95}, it is shown in
\cite{Chang94} that the probability of the delay $D(t)$ at time $t$
exceeding a delay bound $D_{\max}$ satisfies:
\begin{eqnarray}
\sup_{t} \textrm{Pr}[ D(t) \geq D_{\max}] \approx \gamma(R)
e^{-\theta(R) D_{\max}},
\end{eqnarray}
where $\gamma(R)=\Pr\left[D(t) \geq 0\right]$ is the probability
that the buffer is nonempty and
$\theta(R)=RE_B^{-1}(R)$ is the QoS exponent (i.e.,
the solution of $E_B(\theta)=R$ multiplied by $R$).
Both $\gamma(R)$ and $\theta(R)$ are functions of
the constant channel capacity $R$. Thus, a source, which has a
common delay bound $D_{\max}$ and can tolerate a delay bound
violation probability of at most $\epsilon$, can be modeled by the
pair $\{\gamma(R), \theta(R)\}$, where the constant channel capacity
should be at least $R$ with $R$ being the solution of $\gamma(R)
e^{-\theta(R) D_{\max}}=\epsilon$.  The intuitive explanation is
that the tail probability that the delay $D(t)$ exceeds
$D_{\max}$ is proportional to the  probability that the
buffer is nonempty and decays exponentially fast as the threshold
$D_{\max}$ increases. The QoS exponent \cite{Wu03} $\theta(R)$ can
be interpreted as the indicator of the QoS requirement,
i.e., a smaller $\theta(R)$ corresponds to a {\em
looser} QoS requirement and vice versa. As a result, the {\em
effective bandwidth} is defined as the minimum service rate required
by a given arrival process for which the QoS exponent requirement is
satisfied.

Inspired by the effective bandwidth theory, where
the constant service rate $R$ is used in the source traffic modeling
in wired networks, the authors in \cite{Wu03} use the constant
source traffic rate $\lambda$ to model a wireless
communication channel. They propose the {\em
effective capacity}, which is the dual of the {\em effective
bandwidth}. Let $ S(t) \triangleq \sum_{\tau = 1}^{t} R(\tau)$
represent the amount of service (in number of bits)
over the time interval $[0,t)$. Assume that the
Gartner-Ellis limit of $S(t)$, expressed as
$\Lambda_C(\theta)=\lim_{t\to\infty}\frac{\log \mathbf E[e^{-\theta
S(t)}]}{t}$ exists for all $\theta\geq0$. Then, the {\em effective
capacity} function of $S(t)$ is defined as
\begin{align}
E_C(\theta)=-\frac{\Lambda_C(\theta)}{\theta}=-\lim_{t\to\infty}\frac{1}{\theta
t}\log \mathbf E\left[e^{-\theta S(t)}\right].\label{eqn:E-C}
\end{align}
If we further assume the process $\{R(t) \}$ is
uncorrelated, then the effective
capacity reduces to
\begin{eqnarray}
E_C\big(\theta\big) = - \frac{1}{\theta} \log \left( \mathbf E
\{e^{-\theta R(t)} \} \right).\label{eqn:E-C-reduced}
\end{eqnarray}
Consider a queue of infinite buffer size served by a data source of
constant data rate $\lambda$ (in number of bits).
Similar to the {\em effective bandwidth} case, it
is shown in \cite{Wu03} that the probability of the
delay $D(t)$ at time $t$ exceeding a delay bound $D_{\max}$
satisfies:
\begin{eqnarray}
\sup_{t} \textrm{Pr}[ D(t) \geq D_{\max}] \approx \gamma(\lambda)
e^{-\theta(\lambda) D_{\max}},
\end{eqnarray}
where $\gamma(\lambda)=\Pr[D(t) \geq 0]$ is the probability
that the buffer is nonempty and the QoS exponent is
$\theta(\lambda)=\lambda E_C^{-1}(\lambda)$.  Both
$\gamma(\lambda)$ and $\theta(\lambda)$ are
functions of the constant source rate $\lambda$. Thus, a source,
which has a common delay bound $D_{\max}$ and can tolerate a
delay bound violation probability of at most $\epsilon$, can be
modeled by the pair $\{\gamma(\lambda), \theta(\lambda)\}$, where
the constant data rate should be at most $\lambda$ with $\lambda$
being the solution of $\gamma(\lambda) e^{-\theta(\lambda)
D_{\max}}=\epsilon$. Therefore, as the dual of the {\em effective
bandwidth}, the {\em effective capacity} is defined as the maximum
constant arrival rate that a given service process can support in
order to guarantee a QoS requirement specified by $\theta$.

With the above observation, we can incorporate the QoS requirement
into a pure PHY layer requirement. By
interpreting $\theta$ as the QoS constraint, the throughput
maximization problem subject to the delay QoS constraint (in terms
of the exponential tail probability of the queue length
distribution) can be directly transformed into an
effective capacity maximization problem for a given
QoS exponent $\theta$. This approach is widely used when the QoS
constraint is specified in terms of the QoS exponent
$\theta$. Interested readers can refer to
\cite{Tang0705,Tang0708,Tang0806} and references therein for more
detailed descriptions.

A more careful treatment of the average delay constraint is
developed in \cite{Hui07}, where packet flow model is considered.
The principle idea behind this approach is to establish the
relationship among the average delay
requirement $\overline{D}$, the average arrival rate $\overline \lambda$ and
the average service rate $\overline \mu$  using the queueing theory framework
\cite{Kleinrock75}. The following  procedures are performed
to obtain this relationship.
\begin{enumerate}
\item{Express the average system delay in terms of the average residue service time and the average queueing delay.}
\item{Express the average residue service time in terms of the moments of the service process.}
\item{Establish the relationship among the average delay
requirement, the average arrival rate and
the average service rate.}
\end{enumerate}

\begin{Exam}[Equivalent Rate Constraint Approach for Uplink OFDMA Systems]
In the uplink OFDMA system example, we assume the buffer size for each
link is infinite (as in \cite{Hui07}). Hence, the optimization
problem (\ref{eqn:catIII}) can be simplified as follows
\cite{Hui07}:
\begin{eqnarray}
\min_{\Omega = \{ \Omega_{l,s}, \Omega_{l,p} \}} &&
\mathbf{E}^{\Omega} \left[ \sum_{l=1}^L \sum_{m=1}^{N_F}
p_{l,m} \right] \label{eqn:ch3_prob1}\\
\textrm{s.t.} &&\quad  s_{l,m}\in\{0,1\}, \forall l\in \mathcal L, m\in \{1,2,\cdots, N_F\},\quad \sum_{l=1}^L s_{l,m} =1, \quad  \forall l \in \mathcal L \label{eqn:ch3_subc-constraint}\\
&& \mathbf{E}^{\Omega}\left[\sum_{m=1}^{N_F}p_{l,m}\right]\leq P_l, \quad
\forall l\in\mathcal L \label{eqn:ch3_pwr-constraint}\\
&& \overline{\lambda}_l \mathbf{E}^{\Omega}\Big [Q_l(t)\Big] \leq D_l, \quad
\forall l\in \mathcal L. \label{eqn:ch3_delay-constraint}
\end{eqnarray}
Following the standard procedures shown above, the average delay
constraint \eqref{eqn:ch3_delay-constraint} can be replaced with an
equivalent average rate constraint given by \cite[Lemma 1]{Hui07}
\begin{eqnarray}
\mathbf E^{\Omega} \left[ \sum_{m=1}^{N_F} s_{l,m} \log_2 (1 +
p_{l,m} |H_{l,m}|^2) \right] \geq \frac{(2
D_{l}\overline{\lambda}_{l}+ 2)+ \sqrt{(2
D_{l}\overline{\lambda}_{l}+ 2)^2 - 8
D_{l}\overline{\lambda}_{l}}}{4 D_{l}} \overline{N}_l,
\label{eqn:ch3_constraint4}
\end{eqnarray}
where $\overline{N}_l$ and $\overline{\lambda}_l$ are the average
packet size and the average arrival rate of the $l$-th link,
respectively. Applying the above results, the
original optimization problem \eqref{eqn:ch3_prob1} can be
reformulated as follows:
\begin{eqnarray}
\min_{\Omega = \{ \Omega_{l,s}, \Omega_{l,p} \}} &&
\mathbf{E}^{\Omega} \left[ \sum_{l=1}^L \sum_{m=1}^{N_F}
p_{l,m} \right] \label{eqn:ch3_prob2}\\
\textrm{s.t.} && \eqref{eqn:ch3_subc-constraint},
\eqref{eqn:ch3_pwr-constraint},
\eqref{eqn:ch3_constraint4}.\label{eqn:ch3_prob2-const}
\end{eqnarray}
Optimization problem
\eqref{eqn:ch3_prob2}-\eqref{eqn:ch3_prob2-const} is a mixed
combinatorial (w.r.t. integer variables $\{s_{l,m}\}$) and
convex optimization problem (w.r.t. $\{p_{l,m}\}$). If we
relax the integer constraint $s_{l,m} \in \{0,1\}$ into real values,
i.e., $s_{l,m} \in [0,1]$, the resultant problem
\eqref{eqn:ch3_prob2} would be a convex maximization problem. Using
standard Lagrange Multiplier techniques, we can derive the optimal
subcarrier and power allocation as follows:
\begin{align}
s_{l,m} &= \left\{
\begin{array}{ll} 1, &
\text{if} \quad X_{l,m}=\max_{j} \big\{X_{j,m}\big\} >0\\
0, & \textrm{otherwise}
\end{array} \right. \label{eqn:ch_3_joint-subc-allo}\\
p_{l,m}&=s_{l,m}\left(\frac{1+\nu_l}{\gamma_l}-\frac{1}{|H_{l,m}|^2}\right)^{+},\label{eqn:ch_3_joint-pwr-allo}
\end{align}
where $$X_{l,m}=(1+\nu_l)
\log_2\left(1+|H_{l,m}|^2\left(\frac{1+\nu_l}{\gamma_l}-\frac{1}{|H_{l,m}|^2}\right)^{+}\right)-\gamma_l
\left(\frac{1+\nu_l}{\gamma_l}-\frac{1}{|H_{l,m}|^2}\right)^{+},$$
$\gamma_l$ is the Lagrange multiplier corresponding to the average
power constraint in \eqref{eqn:ch3_pwr-constraint} and $\nu_l$ is
the Lagrange multiplier corresponding to the transformed average
rate constraint \eqref{eqn:ch3_constraint4} for the $l$-th
MS.~\hfill\QEDclosed \label{exam:rate}
\end{Exam}

This approach provides  potentially simple
solutions for  single-hop wireless networks in the
sense that the cross-layer optimization problem is transformed into
a purely information theoretical optimization problem. Then, the
traditional PHY layer optimization approach, such as
power allocation and subcarrier allocation, can be
readily applied to solve the transformed problem.
The optimal control policy is a function of the CSI with
some weighting shifts by the delay requirements and
hence it is simple to implement in practical communication systems.
However, the implicit assumption behind this approach is that the
user traffic loading is quite high, or
equivalently, the probability that a buffer is empty is quite low
(i.e., the large delay regime). For general
delay regime, the delay-optimal control policy should be adaptive to
both the CSI and the QSI and the performance of the ``equivalent rate
constraint" approach is not promising, as we shall
illustrate  in Section \ref{sec_comparison}. In addition, due to the
complex coupling among the queues in  multi-hop wireless networks,
it is difficult to express delay constraints in terms of all the
control actions, including routing. Therefore, this approach
cannot be easily extended to multi-hop wireless
networks.

\section{Stochastic Lynapnov Stability Drift Approach}
\label{sec_Lynapnov}

Another important method to deal with delay-aware resource control
in  wireless networks is to directly analyze the
characteristics of the control policies in the stochastic stability
sense using the {\em Lyapunov drift}  technique. The Lyapunov drift
theory has a long history in the field of discrete stochastic
processes and Markov chains
\cite{MCandStochasticStability:2009,AppliedProbQueue:2003}. The
authors of \cite{Tassiulas-Ephremides:1992-2} first applied the
Lyapunov drift theory to develop a general algorithm which
stabilizes a multi-hop packet radio network with configurable link
activation sets. The concepts of maximum weight matching (MWM) and
differential backlog scheduling, developed in
\cite{Tassiulas-Ephremides:1992-2}, play important roles in the
dynamic control strategies in queueing networks. The Lyapunov drift theory is then extended to the Lyapunov optimization theory. In this section, we
first introduce the preliminaries and
the main results on the Lynapnov
stability analysis; after that, we present two
examples, one for the {\em Lyapunov drift} theory and the other for the
{\em Lyapunov optimization} thoery.

\subsection{What is Queue Stability?}

First, we introduce the definition of queue
stability as follows.
\begin{Def}[Queue Stability] $\quad$
\begin{enumerate}
\item A single queue is strongly stable if $ \limsup\limits_{T\rightarrow +\infty}\frac{1}{T}\sum\limits_{t=1}^{T}
\mathbf{E}\left [Q(t)\right]<\infty$.
\item A network of queues is strongly stable if all the individual queues of the network are strongly stable.
\end{enumerate}~\hfill\QEDclosed
\end{Def}
A queue is strongly stable if it has a bounded time average backlog.
Throughout this paper, we use the term ``stability'' to refer to
strong stability. Based on the definition of stability, we define
the stability region as follows.

\begin{Def}[Stability Region] The stability region $\Lambda_{\Omega}$ of  a policy $\Omega$ is the set of average arrival rate
vectors $\left\{\overline{\lambda_n^{(c)}} |  n\in\mathcal{N},
c\in \mathcal C\right\}$ for which the system is stable under $\Omega$.
The stability region of the system $\Lambda$ is the {\em closure} of
the set of all average arrival rates $\left\{\overline{\lambda_n^{(c)}} |
n\in\mathcal{N}, c\in \mathcal C\right\}$ for which a stabilizing
control policy exists. Mathematically, we have
\begin{eqnarray}
\Lambda = \bigcup_{\Omega \in G} \Lambda_{\Omega},
\end{eqnarray}
where $G$ denotes the set of all stabilizing feasible control
policies. ~ \hfill\QEDclosed
\end{Def}

\begin{Def}[Throughput-Optimal Policy]
A {\em throughput-optimal policy} dominates\footnote{A policy
$\Omega_1$ {\em dominates} another policy $\Omega_2$ if
$\Lambda_{\Omega_2} \subset \Lambda_{\Omega_1}$} any other policy in
$G$, i.e. has a stability region that is a superset
of the stability region of any other policy in $G$. Therefore, it
should have a stability region equal to $\Lambda$.
~ \hfill\QEDclosed
\end{Def}

In other words, throughput-optimal policies ensure that the queueing
system is stable as long as the average  arrival rate vector is within the
system stability region. Three classes of policies
that are known to be throughput-optimal are the Max
Weight rule (also known as M-LWDF/M-LWWF \cite{StolyarMLWDF:2004} in
single-hop wireless queueing systems), the Exponential (EXP) rule
\cite{EXPruleShakkottai-Stolyar:2002} and the Log
rule\cite{logrule:2009}. The throughput-optimal property of Max
Weight type algorithms is proved by the theory of {\em Lyapunov
drift}\cite{StolyarMLWDF:2004}, which is introduced
in the next part. The Log rule is also proved to be
throughput-optimal by the theorem related to the
{\em Lyapunov drift} in \cite{logrule:2009}. On the other hand, the
EXP rule is proved to be throughput-optimal by the
{\em fluid limit} technique along with a {\em separation of time
scales} argument in \cite{EXPruleShakkottai-Stolyar:2002}.

\subsection{Main Results on Lyapunov Drift}

In order to show the stability property of the queueing
systems, we rely on the well-developed
stability theory in Markov Chains using
\emph{negative Lyapunov drift}
\cite{Asmussen87,Meyn93,Kumar96,McKeown-M-A-W:1999,Leonardi:INFOCOM:2001}.
We use the quadratic {\em Lyapunov function} $L(\mathbf Q) =
\sum_{n \in \mathcal N,c\in \mathcal C} \big(Q_n^{(c)}\big)^2$  for the system queue
state $\mathbf Q$
through the rest of the paper. Based on the {\em Lyapunov function},
we define the (one-slot) {\em Lyapunov drift} as the expected change
in the Lyapunov function from one slot to the next, which is given
by
\begin{align}
\Delta(\mathbf Q(t))\triangleq \mathbf E \left[L(\mathbf
Q(t+1))-L(\mathbf Q(t))| \mathbf Q(t)
\right].\label{eqn:ch4_onestep_LD}
\end{align}
Therefore, the theory of Lyaponov stability is summarized as
follows\cite{Neely03,Georgiadis-Neely-Tassiulas:2006}:
\begin{Thm}[Lyapunov Drift]
If there are positive values $B$, $\epsilon$ such that for all time
slot $t$ we have:
\begin{align}
\Delta(\mathbf Q(t))\leq B-\epsilon \sum_{n \in \mathcal N,c\in
\mathcal C} Q_n^{(c)}(t), \label{eqn:ch4_LD-condition}
\end{align}
then the network is stable, and the average queue length satisfies:
\begin{align}
\limsup_{T\to \infty}\frac{1}{T}\sum_{t=1}^T \sum_{n \in \mathcal
N,c\in \mathcal C} \mathbf E\left[Q_n^{(c)}(t)\right] \leq
\frac{B}{\epsilon}.\nonumber
\end{align} ~ \hfill\QEDclosed
\label{Thm:LD}
\end{Thm}

Note that if the condition in \eqref{eqn:ch4_LD-condition} holds,
then the Lyapunov drift $\Delta (\mathbf Q(t))\leq - \delta$
($\forall \delta>0$) whenever $\sum_{n\in \mathcal N, c\in \mathcal C} Q_n^{(c)}(t) \geq
\frac{B+\delta}{\epsilon}$. Intuitively, this property ensures
network stability because whenever the queue length vector leaves
the bounded region for the sum  queue length, i.e.\textcolor{red}{,}
$\left\{\mathbf Q\succeq0:\sum_{n\in \mathcal N, c\in \mathcal C} Q_n^{(c)}(t) \leq
\frac{B+\delta}{\epsilon}\right\}$, the negative drift $\Delta (\mathbf
Q(t))\leq - \delta$ eventually drives it back to this region.

Next, we illustrate how to use the Lyapunov drift
to prove the stability of queueing networks and develop
stabilizing control algorithms. Define the maximum
input rate and output rate of node $n$ as $\mu_{\max,n}^{in}=\sup_t
\sum_{c\in \mathcal C}\sum_{l\in \{l:d(l)=n\}}\mu_l^{(c)}(t)$ and
$\mu_{\max,n}^{out}=\sup_t\sum_{c\in \mathcal C} \sum_{l\in
\{l:s(l)=n\}}\mu_l^{(c)}(t)$, respectively. They are finite due to
the resource allocation constraints. Assume the
total exogenous arrival to node $n$ is bounded by a constant
$\lambda_n^{\max}=\sup_t \sum_{c\in \mathcal C}\lambda_n^{(c)}(t)$. From the queue
dynamics in \eqref{eqn:sec2:queue}, we have the following bound for
Lyapunov drift\cite{Georgiadis-Neely-Tassiulas:2006}:
\begin{align} \Delta (\mathbf Q(t))\leq&
B+2\sum_{n\in \mathcal N, c\in \mathcal C}
Q_n^{(c)}(t)\overline{\lambda_n^{(c)}}\nonumber\\
 &-2\mathbf
E\left[\sum_{n\in \mathcal N, c\in \mathcal C}Q_n^{(c)}(t)\left(\sum_{l\in
\{l:s(l)=n\}}\mu_l^{(c)}(t)-\sum_{l\in
\{l:d(l)=n\}}\mu_l^{(c)}(t)\right)\Big|\mathbf
Q(t)\right],\label{eqn:ch4_LD_ineq}
\end{align}
where
\begin{align}
B\triangleq \sum_{n\in \mathcal N}\left((\mu_{\max,n}^{out})^2 +
(\lambda_n^{\max}+\mu_{\max,n}^{in})^2\right).\label{eqn:ch4_def_B}
\end{align}
The {\em dynamic backpressure algorithm} (DBP) is  designed to
minimize the upper bound of the Lyapunov drift (the R.H.S. of
\eqref{eqn:ch4_LD_ineq}) over all policies at each time slot.
For  single-hop wireless networks, we use the link
index $l$ to specify each queue instead of the node index $n$ and
the commodity index $c$ for notational simplicity.  From
\eqref{eqn:ch4_LD_ineq}, the single-hop {\em dynamic
backpressure algorithm} (M-LWDF/M-LWWF)  maximizes $\sum_{l\in
\mathcal L} Q_l(t) \mu_l (t)$ under resource allocation constraints.
Based on   Theorem \ref{Thm:LD}, it is shown that the DBP algorithm
is throughput-optimal
\cite{Neely:PhD:2003,Neely-Modiano-Rohrs:2005}.

After the introduction of dynamic control algorithm
designs in \cite{Tassiulas-Ephremides:1992-2}, the
Lyapunov drift approach is successfully used to
optimize the allocation of computer resources
\cite{BhattacharyaAdaptiveLexicographicOpt:1993}, stabilize packet
switch systems \cite{Keslassy-McKeown:2001,Kumar-Meyn:1995,
Leonardi:INFOCOM:2001, McKeown-M-A-W:1999}, satellite and wireless
systems
\cite{Tassiulas-Ephremides:1993-2,Andrew-Kumaran-Stolyar:2001,Neely-Modiano-Rohrs:2003}.
For example, the concepts of MWM and
differential backlog scheduling are first developed
in \cite{Tassiulas-Ephremides:1992-2} based on the Lyapunov drift
theory. Using the linear programming argument and
the Lyapunov drift theory, it is proved that a MWM
algorithm can achieve a throughput of 100$\%$  for both uniform and
nonuniform arrivals in \cite{McKeown-M-A-W:1999}. Based on the
analytical techniques of the Lyapunov drift, the
bounds on the average delay and queue size averages as well as
variances in input-queued cell-based switches under MWM
are derived in \cite{Leonardi:INFOCOM:2001}. By
the Lyapunov drift theory, the Longest Connected
Queue (LCQ) algorithm is proved in
\cite{Tassiulas-Ephremides:1993-2} to stabilize the system under
certain conditions and minimize the delay for the special case of
symmetric queues (i.e., queues with equal arrival,
service and connectivity statistics). Due to page limitation, we
refer the readers to the above references for the details.

\subsection{Main Results on Lynapnov Optimization}

The {\em Lyapunov drift} theory is
extended to the {\em Lyapunov optimization} theory, through which
we can stabilize queueing networks while
additionally optimize some performance metrics and
satisfy additional
constraints\cite{Georgiadis-Neely-Tassiulas:2006}.

Let $\mathbf x(t)=\big(x_1(t),x_2(t),\cdots, x_K(t)\big)$
represent any associated vector control process
that influences the dynamics of the vector queue length
process $\mathbf Q(t)$. Let $g:\mathbb R^K\to \mathbb R$ be any
scalar valued concave function. Define  $\overline {\mathbf
x}(T)\triangleq\frac{1}{T} \sum_{t=1}^T \mathbf E [\mathbf x (t)]$
and
$\overline g \triangleq\limsup_{T\to \infty} \frac{1}{T} \mathbf E
[g(\mathbf x(t))]$.  Suppose  the goal is to stabilize
the $\mathbf Q(t)$ process while maximizing
$g(\cdot)$ of the time average of the $\mathbf
x(t)$ process, i.e., maximizing
$g(\overline{\mathbf x})$, where $\overline{\mathbf x}=\limsup_{T\to
\infty} \overline {\mathbf x}(T)$. Let $g^*$ represent a desired
``target'' utility value.   The theory of Lyapunov optimization is
summarized as follows:
\begin{Thm}[Lyapunov Optimization] If there are positive constants $V,\epsilon, B$ such that  for all
$t$ and all $\mathbf Q(t)$, the {\em  Lyapunov drift} satisfies:
\begin{eqnarray}
\Delta (\mathbf Q(t)) - V \mathbf{E}[ g(\mathbf x(t)) | \mathbf
Q(t)] \leq B - \epsilon \sum_{n\in \mathcal N, c \in \mathcal C}
Q_n^{(c)}(t) - V g^*, \label{eqn:ch4_L-opt-drfit-bound}
\end{eqnarray}
then, we have:
\begin{eqnarray}
\limsup_{T \rightarrow \infty} \frac{1}{T} \sum_{t = 1}^T \sum_{n\in
\mathcal N, c \in \mathcal C} \mathbf E\left[Q_n^{(c)}(t)\right]  & \leq &
\frac{B +
V (\overline g -g^*)}{\epsilon},\\
\liminf_{T \rightarrow \infty} \frac{1}{T} \sum_{t = 1}^T
g\left(\overline{\mathbf x}(t)\right)& \geq & g^* - \frac{B}{V}.
\end{eqnarray} ~ \hfill\QEDclosed\label{Thm:Lyap_opt}
\end{Thm}
Note that a similar result can be shown for minimizing a convex
function $h:\mathbb R^K\to \mathbb R$ by defining
$g(\cdot)=-h(\cdot)$ and reversing inequalities where appropriate.

Theorem \ref{Thm:Lyap_opt} is most useful when the quantity
$\overline g -g^*$ can be bounded by a constant. Specifically, if
$\overline g -g^*\leq G_{\max}$, then the lower bound on the achieved utility $\overline
g$ can be pushed arbitrarily close to the target utility $g^*$ with
a corresponding increase (linear in $V$) in the upper bound on
$\limsup_{T \rightarrow \infty} \frac{1}{T} \sum_{t = 1}^T\sum_{n\in
\mathcal N, c \in \mathcal C} \mathbf E[Q_n^{(c)}(t)]$. The Lyapunov
optimization theorem in Theorem \ref{Thm:Lyap_opt} suggests that a
good control strategy is to greedily minimize the following drift
metric at every time slot
\[\Delta (\mathbf Q(t)) - V \mathbf{E}\left[ g(\mathbf x(t)) | \mathbf
Q(t)\right], \] i.e.\textcolor{red}{,} the L.H.S. of
\eqref{eqn:ch4_L-opt-drfit-bound}.

Next, we introduce the {\em Energy-Efficient Control Algorithm}
(EECA) \cite{Neely:2006}, which utilizes the {\em Lyapunov
optimization} theory to develop an algorithm that stabilizes the
system and consumes an average power that is arbitrarily close to the
minimum power solution. From
\eqref{eqn:ch4_LD_ineq}, we have
\begin{align}
&\Delta (\mathbf Q(t))+V \mathbf{E}\left[ \sum_{n\in \mathcal N} P_n(t)\Big |
\mathbf Q(t)\right]\leq B+2\sum_{nc}
Q_n^{(c)}(t)\overline{\lambda_n^{(c)}}\nonumber\\
&- \mathbf E\left[2\sum_{nc}Q_n^{(c)}(t)\left(\sum_{l\in
\{l:s(l)=n\}}\mu_l^{(c)}(t)-\sum_{l\in
\{l:d(l)=n\}}\mu_l^{(c)}(t)\right)-V\sum_{n\in \mathcal N}
P_n(t)\Bigg|\mathbf Q(t)\right].\label{eqn:ch4_delay-pwr_ineq}
\end{align}
EECA is designed to minimize the R.H.S. of the
inequality in \eqref{eqn:ch4_delay-pwr_ineq} over all possible power
allocation strategies. For single-hop wireless networks, we use link index $l$ to denote the node index as well as the commodity index. From
\eqref{eqn:ch4_delay-pwr_ineq}, we have that the
single-hop EECA maximizes $\sum_{l\in \mathcal L} \big(2Q_l(t) \mu_l
(t)-V P_l(t)\big)$ over all possible power allocation strategies at
each slot $t$. Based on
Theorem \ref{Thm:Lyap_opt}, it is shown that the EECA is
throughput-optimal and can achieve $[\mathcal O(1/V), \mathcal
O(V)]$ power-delay tradeoff by adjusting the parameter $V$
\cite{Georgiadis-Neely-Tassiulas:2006}.

The Lyapunov optimization theory also has applications in
the transport layer flow control and network
fairness optimization when the exogenous arrival rates are outside
of the network stability region. The Cross Layer
Control (CLC) algorithm is designed to greedily
minimize the R.H.S. of \eqref{eqn:ch4_L-opt-drfit-bound} to achieve
a utility of exogenous arrival rates, which is arbitrarily close to
optimal while maintaining  network
stability\cite{Neely:PhD:2003,Neely-Modiano-Rohrs:2005,Georgiadis-Neely-Tassiulas:2006}.

\begin{Rem}
Note that the average delay bounds developed in
Theorem \ref{Thm:LD} and Theorem \ref{Thm:Lyap_opt} are tight only
when the traffic loading is sufficiently high, while the tightness
of such delay bounds for moderate and light traffic
loading is not known.~\hfill\QEDclosed
\end{Rem}

\subsection{Methodology and Example}
To stabilize  queueing networks
using the Lyapunov drift theorem in Theorem
\ref{Thm:LD} or stabilize  queueing networks while
additionally optimizing some performance metrics
(e.g., maximizing the average weighted sum system
throughput in \eqref{eqn:catI}, minimizing the average weighted sum
power in \eqref{eqn:catIII}, etc) using the Lyapunov optimization
theorem in Theorem \ref{Thm:Lyap_opt}, the procedure can be
summarized as follows:
\begin{enumerate}
\item Choose a Lyapunov function and calculate the Lyapunov drift $\Delta (\mathbf Q(t))$ or $\Delta (\mathbf Q(t)) - V \mathbf{E}[ g(\mathbf x(t)) | \mathbf
Q(t)]$, where $g(\overline{\mathbf x})$ is the utility to be
maximized.
\item Based on the system state observations, minimize the upper bound of $\Delta (\mathbf Q(t))$ or $\Delta (\mathbf Q(t)) - V \mathbf{E}[ g(\mathbf x(t)) | \mathbf
Q(t)]$ over all polices at each time slot.
\item Transform other average performance constraints into queue
stability problems using the technique of virtual cost queues
\cite{Neely-Modiano-Rohrs:2005,Georgiadis-Neely-Tassiulas:2006} if
needed\footnote{Please refer to
\cite{Neely-Modiano-Rohrs:2005,Georgiadis-Neely-Tassiulas:2006} for
the details of the virtual cost queue technique,
which we omit   here due to paper limitation. In all the examples, we use the
Lagrangian techniques to deal with these average
constraints to facilitate the derivation of the closed-form solution
to obtain certain insights and comparisons of the solutions obtained
by different techniques.}.
\end{enumerate}

In the following, we  illustrate how to apply the
Lynapnov drift approach and the Lynapnov optimization approach in
resource allocation for uplink OFDMA systems,
respectively.

\begin{Exam}[Lynapnov Drift Approach for Uplink OFDMA Systems]
In the  uplink OFDMA system example, the dynamic
backpressure algorithm  under the subcarrier allocation constraints in
\eqref{eqn:ch3_subc-constraint} and the average power constraints in
\eqref{eqn:ch3_pwr-constraint}, can be obtained by solving the following optimization problem
\begin{align}
\max_{\Omega = \{ \Omega_{l,s}, \Omega_{l,p} \}}&\  \sum_{l=1}^L
Q_l(t)\sum_{m=1}^{N_F} R_{l,m}(t), \quad \forall t \label{eqn:ch4_LD_exam_obj}\\
s.t. & \ \eqref{eqn:ch3_subc-constraint},
\eqref{eqn:ch3_pwr-constraint} \quad \text{are satisfied}.
\end{align}
Similar to Example \ref{exam:rate},
by applying continuous relaxation
(i.e.\textcolor{red}{,} $s_{l,m} \in [0,1]$) and standard convex
optimization techniques, we can derive the optimal subcarrier and
power allocation as follows:
\begin{align}
s_{l,m} &= \left\{
\begin{array}{ll} 1, &
\text{if} \quad X_{l,m}=\max_{j} \big\{X_{j,m}\big\} >0\\
0, & \textrm{otherwise}
\end{array} \right. \label{eqn:ch_3_joint-subc-allo}\\
p_{l,m}&=s_{l,m}\left(\frac{Q_l(t)}{\gamma_l}-\frac{1}{|H_{l,m}|^2}\right)^{+},\label{eqn:ch_3_joint-pwr-allo}
\end{align}
where $$X_{l,m}=Q_l(t)
\log_2\left(1+|H_{l,m}|^2\left(\frac{Q_l(t)}{\gamma_l}-\frac{1}{|H_{l,m}|^2}\right)^{+}\right)-\gamma_l
\left(\frac{Q_l(t)}{\gamma_l}-\frac{1}{|H_{l,m}|^2}\right)^{+},$$ and
$\gamma_l$ is the Lagrange multiplier corresponding to the average
power constraint in \eqref{eqn:ch3_pwr-constraint} for the $l$-th
MS.~\hfill\QEDclosed\label{Exm:LD}
\end{Exam}

\begin{Exam}[Lynapnov Optimization Approach for Uplink OFDMA Systems]
In the uplink OFDMA system
example, the average sum power minimization in \eqref{eqn:catIII}
(with all weights 1 for illustration simplicity) under network
stability constraint and the subcarrier allocation constraints in
\eqref{eqn:ch3_subc-constraint} based on  Theorem \ref{Thm:Lyap_opt} is given by
\begin{align}
\max_{\Omega = \{ \Omega_{l,s}, \Omega_{l,p} \}}&\  \sum_{l=1}^L
\left(
2 Q_l(t)\sum_{m=1}^{N_F} R_{l,m}(t)-V\sum_{m=1}^{N_F} p_{l,m} \right), \quad \forall t \label{eqn:ch4_L-opt_exam_obj}\\
s.t. & \ \eqref{eqn:ch3_subc-constraint} \quad \text{is satisfied}.
\end{align}
Similar to the previous example, we can derive the optimal
subcarrier and power allocation as follows:
\begin{align}
s_{l,m} &= \left\{
\begin{array}{ll} 1, &
\text{if} \quad X_{l,m}=\max_{j} \big\{X_{j,m}\big\} >0\\
0, & \textrm{otherwise}
\end{array} \right. \label{eqn:ch_3_joint-subc-allo}\\
p_{l,m}&=s_{l,m}\left(\frac{2Q_l(t)}{V
}-\frac{1}{|H_{l,m}|^2}\right)^{+},\label{eqn:ch_3_joint-pwr-allo}
\end{align}
where $$X_{l,m}=2Q_l(t)
\log_2\left(1+|H_{l,m}|^2\left(\frac{2Q_l(t)}{V
}-\frac{1}{|H_{l,m}|^2}\right)^{+}\right)-V \left(\frac{2Q_l(t)}{V
}-\frac{1}{|H_{l,m}|^2}\right)^{+}.$$ Note that
the parameter $V$ is used to
adjust the average power-delay tradeoff.~\hfill\QEDclosed
\label{Exm:L-opt}
\end{Exam}

\begin{Rem}
The {\em Lyapunov stability drift approach} provides a
simple alternative to deal with delay-aware control
problems. The derived {\em cross-layer control policies}
are adaptive to both the CSI and the QSI. The derived policies are also
throughput-optimal (in stability sense). However, as we shall
illustrate, throughput optimality (stability) is only a weak form of
delay performance and the derived policies may not have good delay
performance especially in the small delay regime.
There are many recent studies focusing on delay
reduction in the traditional DBP algorithm in
multi-hop networks and we shall further elaborate this in Section
\ref{sec_routing}.~\hfill\QEDclosed
\end{Rem}

\section{Markov Decision Process and Stochastic Learning Approach}
\label{sec_MDP} In wireless networks, the system state can be
characterized by the aggregation of the  CSI and the QSI.  In fact, under  Assumptions
\ref{ass:channel} and \ref{ass:arrival}, the system state dynamics
evolves as a controlled Markov chain and the
delay-optimal resource control can be modeled as an infinite horizon
average cost MDP \cite{Bertsekas:2007}. MDP is a systematic
approach for delay-optimal control problems, which in general could give
optimal solutions for any operating
regime. However, the main issue associated with the MDP approach is
the {\em curse of dimensionality}. For instance, the cardinality of
the system state space is exponential w.r.t. the number of queues in
the wireless network and hence solving the MDP is quite complicated
in general. In addition, the optimal control actions
are adaptive to the global system QSI and CSI but in some cases,
these CSI and QSI observations are obtained locally at each
node. Hence, a brute-force
centralized solution will lead to enormous complexity as well as
signaling loading to deliver the global CSI and QSI to the
controller. In this section, we briefly summarize the key theories
of MDP and stochastic approximation (SA) and illustrate how we could
utilize the techniques of {\em approximate MDP} and {\em stochastic
learning} to overcome the complexity as well as the
distributed implementation requirement in delay-aware resource control.

\subsection{Why Delay-Optimal Control is an MDP?}

In general, an MDP can be characterized by four
elements, namely the {\em state space}, the {\em action space}, the
{\em state transition probability} and the {\em system cost},
which are defined as follows:
\begin{itemize}
\item $\mathcal X=\{\chi^1,\chi^2,\cdots\}$: the finite state space with $|\mathcal{X}|$ states;

\item $\mathcal{A}=\{a^1,a^2,\cdots \}$: the  action space;

\item $\Pr[\chi'|\chi,a]$: the transition probability from state $\chi$ to state $\chi'$ under action
$a$; and

\item $g(\chi,a)$: the system cost in state $\chi$ under action $a$.
\end{itemize}
Therefore, an MDP is a  4-tuple
$\big(\mathcal{X,A},\Pr[\cdot |\cdot,\cdot], g(\cdot, \cdot)\big)$.
A {\em stationary and deterministic control policy} $\Omega:
\mathcal X\to \mathcal A$ is a mapping from the state space $\mathcal{X}$ to the action space $\mathcal{A}$,
which determines the specific action taken when the system is in
state $\chi$. Given  policy $\Omega$, the corresponding random
process of the system state and the per-stage cost $\big(\chi(t),
g(t)\big)$ evolves as a Markov chain with the
probability measure induced by the transition kernel
$\Pr[\chi'|\chi,\Omega(\chi)]$.
The goal of the infinite horizon average cost problem is to find an optimal policy
$\Omega^*$ such that the long term average cost is minimized
among all feasible policies, i.e.,
\begin{equation}
\min_{\Omega} \limsup_{T \rightarrow \infty} \frac{1}{T}
\sum_{t=1}^{T} \mathbf{E}^{\Omega} \left[
g\left(\chi(t),\Omega\left(\chi(t)\right)\right) \right],\nonumber
\end{equation}
where  $\mathbf{E}^{\Omega}$ denotes the expectation operator taken
w.r.t. the probability measure induced by the control policy
$\Omega$. If the set of feasible policies are unichain policies, then
the optimization problem can be written as
\begin{equation}
\min_{\Omega} \limsup_{T \rightarrow \infty} \frac{1}{T}
\sum_{t=1}^{T} \mathbf{E}^{\Omega} \left[
g\left(\chi(t),\Omega(\chi(t))\right) \right] = \min_{\Omega}
\mathbf{E}^{\pi(\Omega)} \left[ g \left( \chi, \Omega(\chi) \right)
\right],\nonumber
\end{equation}
where $\pi(\Omega)$ is the unique steady state distribution given policy
$\Omega$.

Assume the buffer size is finite and denoted as $N_Q$. The system
queue dynamics $\mathbf{Q}(t)$ evolves according to
\eqref{eqn:sec2:queue} with projection onto $[0,N_Q]$  and the
arrival, departure and the CSI processes are Markovian under
Assumptions \ref{ass:channel} and
\ref{ass:arrival}. Hence, the system state $\chi(t)$ is a finite
state {\em controlled Markov chain} with the following
correspondence:
\begin{itemize}
\item The system state space in the delay-optimal control problem  is defined as the aggregation of
the system CSI and the system QSI, thus $\mathcal X = \mathcal
H\times \mathcal Q$. $\mathcal Q$ and $\mathcal X$ are both finite.

\item The action space is the space of all control actions, including  the resource allocation
actions (e.g., power allocation actions and subcarrier allocation actions) and the routing actions.

\item The transition kernel is given by
\begin{eqnarray}
\Pr\big[\chi(t+1)|\chi(t),\Omega(\chi(t))\big]  =
\Pr\big[\mathbf{Q}(t+1)|\chi(t),\Omega(\chi(t))\big]
\Pr\big[\mathbf{H}(t+1)|\mathbf{H}(t)\big]. \nonumber
\end{eqnarray}

\item By the standard Lagrangian approach, the optimization problem in (\ref{eqn:catII}) can be transformed as follows
\begin{equation} \min_{\Omega} \
L^{\Omega}=\limsup_{T \rightarrow \infty} \frac{1}{T} \sum_{t=1}^{T}
\mathbf{E}^{\Omega}  \left[\sum_{n \in \mathcal{N}} \sum_{c \in
\mathcal{C}} \left( \nu_n^{(c)} Q_n^{(c)}(t)+ \eta_n^{(c)}
\mathbf{I}\left[Q_n^{(c)}(t)=N_Q\right]\right)+\sum_{n \in \mathcal{N}}\gamma_{n}
P_n(t) \right] \nonumber,
\end{equation}
where $\gamma_{n}$ and $\eta_n^{(c)} $ are Lagrange multipliers
corresponding to the average power and average drop rate constraints.
Therefore, the per-stage system cost  function given a system state
$\chi$ is defined as
\begin{equation}
g\big(\chi,\Omega(\chi)\big) = \sum_{n \in \mathcal{N}} \sum_{c \in
\mathcal{C}} \big( \nu_n^{(c)} Q_n^{(c)}+ \eta_n^{(c)}
\mathbf{I}[Q_n^{(c)}=N_Q]\big)+\sum_{n \in \mathcal{N}}\gamma_{n}
P_n \label{eqn:ch5_reward},
\end{equation}
\end{itemize}
As a result, there is a one-one correspondence between the
delay-optimal control problem and the MDP. The average delay
minimization problem under average power constraint
in  (\ref{eqn:catII}) can be modeled as an infinite
horizon MDP to minimize the average cost (delay) per stage as
follows
\begin{Prob}[Delay-Optimal MDP Formulation]\label{Prob:ch5_delay_opt_MDP}
\begin{equation}
\min_{\Omega}L^{\Omega}= \min_{\Omega} \mathop {\limsup}\limits_{T
\to \infty } \frac{1}{T}\sum\limits_{t = 1}^T {\mathbf{E}^{\Omega}
\left[ {g\left( {\chi \left( t \right), \Omega\left( {\chi \left( t
\right)} \right)} \right)} \right]}. \label{eqn:ch5_MDP_opt_prob}
\end{equation}~ \hfill\QEDclosed
\end{Prob}

Note that we restrict our policy space to the unichain policies
where the {\em induced} Markov chains under all feasible unichain
policies are ergodic and share the same state space $\mathcal{X}$.
In addition, we assume the induced Markov chains are irreducible and
hence, the chains are ergodic with steady state distribution
$\pi(\Omega)$.  In this case, the limit of infinite horizon average
cost under policy $\Omega$ (i.e., $L^{\Omega}$)
exists with probability 1 (w.p.1) and is independent of the initial
state.

\subsection{Optimal Solution of the Delay-Optimal MDP}

Under the unichain policy space  assumption, the delay-optimal control policy
of the above MDP is given by the solution of the Bellman
equation\cite{Bertsekas:2007}. This is summarized in the following
Lemma.
\begin{Lem}[Bellman Equation]
If a scalar $\theta$ and a vector
$\mathbf{V}=\left[V(\chi^1),V(\chi^2),\cdots\right]$ satisfy the Bellman
equation for the delay-optimal MDP in Problem
\ref{Prob:ch5_delay_opt_MDP}:
\begin{equation}
\theta  + V ( {\chi ^i } ) = \min_{\Omega} \left\{ g\left(
{\chi^i,\Omega( {\chi ^i })} \right) + \sum\limits_{\chi ^j } {\Pr
\left[ \chi ^j |\chi ^i ,\Omega( \chi ^i  ) \right] V( \chi ^{\rm{j}}
)}  \right\}, \quad \forall \chi^i\in \mathcal
X,\label{eqn:ch_5_bellman}
\end{equation}
then $\theta$ is the optimal average cost per stage
\begin{equation}
\theta  = \mathop {\min }\limits_\Omega  L^\Omega  = L^*.
\end{equation}
Furthermore, if $\Omega^*$ attains the minimum in
\eqref{eqn:ch_5_bellman} for any $\chi^i \in \mathcal X$,
it is the optimal control
policy.~\hfill\QEDclosed\label{Lem:ch_5_optimality_bellman}
\end{Lem}

The Bellman equation (\ref{eqn:ch_5_bellman}) can be solved numerically by {\em Offline Relative Value Iteration}\cite{Bertsekas:2007} under
certain conditions. While the general solution of the MDP in
(\ref{eqn:ch5_MDP_opt_prob}) can be expressed as a Bellman equation
in (\ref{eqn:ch_5_bellman}), this is still quite far from getting a
desired solution. There are two major issues, namely the {\em
complexity issue} and the {\em signaling overhead issue}. Although
the relative value iteration approach \cite{Bertsekas:2007} can give optimal solution to
the MDP in (\ref{eqn:ch5_MDP_opt_prob}), the
solution is usually too complicated to compute due to the curse of
dimensionality. For example, consider a wireless network with $N$
queues; the total number of the system QSI states is
$(N_Q + 1)^{N}$ ($N_Q$ is the buffer size of each queue), which grows
exponentially with the number of queues. Thus, it is essentially
impossible to compute the potential function at every possible state
even for wireless networks with a small number of
queues.  Another technical challenge is the distributed
implementation requirement of the control
algorithm. For instance, even if we could obtain the potential
function $\mathbf{V}$, the derived control will be
centralized, requiring knowledge of the global system CSI and QSI  at each time slot. This is highly
undesirable due to the huge signaling overhead. From an
implementation perspective, it is desirable to obtain distributed
solutions where each node computes the control
action based on the local CSI and  QSI only.

In the following sections, we  propose two novel
approaches to address the above complexity and overhead issues using
approximate MDP and stochastic learning. We first briefly
summarize some major preliminary results on stochastic approximation
\cite{Kushner:2003,Borkar:08}. The stochastic approximation
algorithm considered in this paper can be characterized by the
following $d$-dimensional recursion
\begin{equation}
\mathbf{X}_{n+1} = \mathbf{X}_n + \epsilon_n \big[\mathbf
h(\mathbf{X}_n) + \mathbf{Z}_n\big], \label{eqn:x}
\end{equation}
where $\mathbf{X}_n =\left [X_n(1), X_n(2),\cdots, X_n(d)\right]^T$ is a
$d$-dimension vector, and $\{\epsilon_n\}$ is a
sequence of positive step
sizes. If the following conditions are satisfied, we have Theorem \ref{thm:sa} on the convergence property (Theorem 2, \cite{Borkar:08}).
\begin{itemize}
\item The map $h:\mathbb R^d\rightarrow \mathbb R^d$ is Lipschitz: $||\mathbf h(x)-\mathbf h(y)|| \leq L ||x-y||$, for $0<L<\infty$,

\item $\sum_n \epsilon_n = \infty, \quad \sum_n \epsilon_n^2 < \infty$,

\item $\{\mathbf{Z}_n\}$ is a Martingale difference sequence w.r.t. the increasing family of $\sigma$-field:
\begin{equation}
\mathcal{F}_n = \sigma(\mathbf{X}_m,\mathbf{Z}_m,m \leq n)\nonumber.
\end{equation}
Furthermore, $\{\mathbf{Z}_n\}$ are square-integrable with
\begin{equation}
\mathbf{E}\left[||\mathbf{Z}_{n+1}||^2\big|\mathcal{F}_n\right] \leq
C(1+||\mathbf{X}_n||^2) \quad a.s., \quad n\geq 0,\nonumber
\end{equation}
for some constant $C>0$.

\item The iterates in \eqref{eqn:x} remain bounded almost surely, i.e.\textcolor{red}{,} $\sup_n ||\mathbf{X}_n|| < \infty$,
a.s..
\end{itemize}

\begin{Thm}
Sequence $\mathbf{X}_n$ generated by (\ref{eqn:x})
converges almost surely to a (possibly sample path dependent) compact connected
internally chain transitive invariant set of the following ordinary
differential equation (ODE)
\begin{equation}
\dot{\mathbf{X}}(t) = h(\mathbf{X}(t)). \nonumber
\end{equation}~\hfill\QEDclosed\label{thm:sa}
\end{Thm}

\subsection{Approach 1: Approximating  Potential Functions}\label{sec_MDP_sub_appro}

In the existing literature on approximate MDP, Bellman and Dreyfus
\cite{Bellman:59} first propose to use polynomials as compact
representations for approximating potential functions. In
\cite{Whitt:78,Reetz:77} the authors discuss different approaches for reducing the size of the state space, which lead to
compact representations of potential functions. On the other hand, in
\cite{Schweitzer:85}, the authors develop several
techniques for approximating potential functions using linear
combinations of fixed sets of basis functions. However, these
approaches are centralized and are only focused on
reducing computational complexity. In this section, we  propose a
novel feature-based method that
addresses both the complexity issue as well as the
distributed requirement.

Similar to Sections
\ref{sec_Rate_Constraint_formulation} and \ref{sec_Lynapnov}, we
consider the uplink OFDMA system example. For notational simplicity, we use the link index $l\in \{1,2,\cdots, L\}$ to denote the
node index as well as the commodity index. Specifically, we have:
\begin{itemize}
\item $\chi_l=(\mathbf{H}_l,Q_l)$ denotes the local state of the $l$-th link. Thus, the global  state $\chi \in \mathcal{X}$ is
the aggregation of the local system states of all
links $\chi = \left\{\chi_l |  l \in \mathcal{L} \right\}$.

\item $\mathcal{X}_l$ denotes the state space of the local states of the $l$-th link. Moreover, the elements in $\mathcal{X}_l$ are enumerated as
$\mathcal{X}_l = \{\chi_l^{\tau}|\tau=1,2,\cdots\}$, where $\tau$
denotes the dummy index enumerating all the local states.

\item Local per-stage cost of the $l$-th link $g_l$ is given by
\begin{equation}
g_l\big(\chi_l,\Omega_l(\chi)\big) = \nu_l Q_l+ \gamma_l
\sum_{m=1}^{N_F} p_{l,m} + \eta_l \mathbf{I} (Q_l=N_Q). \nonumber
\end{equation}
Thus, the overall per-stage cost is given by $g
\big(\chi,\Omega(\chi)\big) = \sum_{l \in \mathcal{L}} g_l
\big(\chi_l,\Omega_l(\chi)\big)$.
\end{itemize}

We consider the linear approximation architecture of the potential
function given below:
\begin{equation}
V(\chi) = \sum_{l=1}^L \sum_{\tau=1}^{|\mathcal{X}_l|}
 \mathbf{I} [\chi_l=\chi^{\tau}_l]\widetilde{V}_l (\chi^{\tau}_l)
\nonumber
\end{equation}
with the vector form given by
\begin{align}
\mathbf{V} = \mathbf{M} \widetilde{\mathbf{V}},
\label{eqn:ch5_vector_aprrox}
\end{align}
where $\{\widetilde{V}_l(\chi_l)\}$
($\forall l \in \mathcal{L}$)  are the {\em per-link
potential functions}, which are defined to be the
solution of the Bellman equation (\ref{eqn:ch_5_bellman}) on some
pre-determined representative system states. We  refer
to these pre-determined subset of system states as
the {\em representative states}. Without loss of generality, we
define the set of the representative states as $\mathcal{X}_R =
\{\chi(l,\tau)|  l \in \mathcal{L},
\tau=2,\cdots,|\mathcal{X}_l|\}$, where $\chi(l,\tau)$ denotes the
joint system state with $\chi_l(l,\tau)=\chi^{\tau}_l$
($2\leq\tau\leq|\mathcal{X}_l|$) and $\chi_{l'}(l,\tau)=\chi_{l'}^1$
( $\forall l' \neq l$). $\mathbf{V} =
[V(\chi^1),\cdots,V(\chi^{|\mathcal{X}|})]^T$ is the vector form of
the original potential function (referred to as
{\em global potential function} in the rest of this paper). The
\textit{parameter vector} $\widetilde{\mathbf{V}}$ and  the
\textit{mapping matrix} $\mathbf{M}$ are given
below:
\begin{align}
\widetilde{\mathbf{V}} &= \bigg[
\widetilde{V}_1(\chi^1_1)\cdots\widetilde{V}_1(\chi_1^{|\mathcal{X}_1|}),\widetilde{V}_2(\chi^1_2)\cdots\widetilde{V}_2(\chi_2^{|\mathcal{X}_2|}),\cdots,
\widetilde{V}_L(\chi^1_L)\cdots\widetilde{V}_L(\chi_L^{|\mathcal{X}_L|})
\bigg]^T \nonumber\\
\mathbf{M} &= \left[ \begin{array}{ccccccc}
    \mathbf{I}[\chi^1_1=\chi^1_1] & \cdots & \mathbf{I}[\chi^1_1=\chi^{|\mathcal{X}_1|}_1], & \cdots & ,\mathbf{I}[\chi^1_L=\chi^1_L] & \cdots & \mathbf{I}[\chi^1_L=\chi^{|\mathcal{X}_1|}_L] \\
    \cdots & \cdots & \cdots, & \cdots &, \cdots & \cdots & \cdots \\
    \mathbf{I}[\chi^{|\mathcal{X}|}_1=\chi^1_1] & \cdots & \mathbf{I}[\chi^{|\mathcal{X}|}_1=\chi^{|\mathcal{X}_1|}_1], & \cdots &, \mathbf{I}[\chi^{|\mathcal{X}|}_L=\chi^1_L] & \cdots & \mathbf{I}[\chi^{|\mathcal{X}|}_L=\chi^{|\mathcal{X}_1|}_L] \nonumber
    \end{array} \right],
\end{align}
where we let
$\widetilde{V}_1(\chi^1_1)=\widetilde{V}_2(\chi^1_2)=\cdots=\widetilde{V}_L(\chi^1_L)=0$
and $\chi^1_l=(\mathbf{H}^1_l,Q^1_l=0)$ ($\forall l \in
\mathcal{L}$). Moreover, we define the inverse mapping matrix
$\mathbf{M}^{-1}$ as
\begin{align}
\mathbf{M}^{-1} = \left[ \begin{array}{ccc}
    \mathbf{I}[\chi^1=\chi(1,1)] \cdots  \mathbf{I}[\chi^1=\chi(1,|\mathcal{X}_1|)], & \cdots, &\mathbf{I}[\chi^1=\chi(L,1)]  \cdots \mathbf{I}[\chi^1=\chi(L,|\mathcal{X}_L|)]\\
    \cdots \quad \quad & \cdots & \cdots \quad \quad\\
    \mathbf{I}[\chi^{|\mathcal{X}|}=\chi(1,1)] \cdots  \mathbf{I}[\chi^{|\mathcal{X}|}=\chi(1,|\mathcal{X}_1|)], & \cdots ,&\mathbf{I}[\chi^{|\mathcal{X}|}=\chi(L,1)]  \cdots  \mathbf{I}[\chi^{|\mathcal{X}|}=\chi(L,|\mathcal{X}_L|)]
    \nonumber
    \end{array} \right]^T.
\end{align}
Thus, we have
\begin{align}
\widetilde{\mathbf{V}} = \mathbf{M}^{-1}
\mathbf{V}.\label{eqn:ch5_vector_approx_inv}
\end{align}

One challenge in utilizing the above approximate MDP is how to
determine the per-link potential function $\widetilde{\mathbf{V}}$.
Instead of solving the Bellman equation on the  representative states, we
estimate $\widetilde{\mathbf{V}}$ using the
stochastic approximation techniques. Specifically, the distributed
online iterative algorithm is given by:
\begin{Alg}[Distributed Online Algorithm for Estimating the Per-Link Potential Functions]$\ $
\begin{itemize}
\item Step 1 (\textbf{Initialization}): Start with a set of initial per-link potential vector
$\widetilde{\mathbf{V}}_{0}$ with $\widetilde{V}_{l,0}(\chi^1_l)=0$
($\forall l \in \mathcal L$).

\item Step 2 (\textbf{Calculate Control Actions}): Based on the realtime observation of the system state $\chi(t)$ at  slot $t$, calculate control actions according to
\begin{equation}
\Omega_t^{*}(\chi(t)) =\arg \min_{\Omega}\left\{\sum_{l \in \mathcal{L}} g_l \left(
\chi_l(t), \Omega_l(\chi(t))
\right)+\sum_{\chi^j}\Pr\left[\chi^j|\chi(t),\Omega(\chi(t))\right]
V_t(\chi^j)\right\} \label{eqn:app_Q_con},
\end{equation}
where $\chi(t)=\{\chi_1(t),\cdots,\chi_L(t)\}$ and $V_t(\chi^j) =
\sum_{l=1}^L \widetilde{V}_{l,t} (\chi^j_l)$.

\item Step 3 (\textbf{Update Per-Link Potential Functions}): After the control policy has been determined, update  all the per-link
potentials functions $\{\widetilde{V}_{l}(\chi^{\tau}_l): 1\leq \tau \leq
|\mathcal X_l|\}$ ($\forall l \in \mathcal L$) based on the
real-time local observations
$\chi_l(t)=\big(\mathbf H_l(t),Q_l(t)\big)$, as follows:
\begin{align}
&\widetilde{V}_{l,t+1}(\chi^{\tau}_l)=\widetilde{V}_{l,t}(\chi^{\tau}_l)+\epsilon_{c(l,\tau,t)}\bigg[\Big(g_l\big(\chi^{\tau}_l,\Omega^{*}_t(\chi(t)\big)+\sum_l
\mathbf E_{\mathbf H_l'}[\widetilde{V}_{l,t}(\mathbf H_l',Q_l(t+1))|\mathbf H_l(t)]\Big)\nonumber\\
&-\Big(g_l\big(\chi^1_l,\Omega^{*}_{\bar
t}(\chi^I)\big)+\sum_l \mathbf E_{\mathbf
H_l'}[\widetilde{V}_{l,t}(\mathbf H_l',Q_l(\bar t+1))|\mathbf
H_l(\bar t)]\Big)-\widetilde{V}_{l,t}(\chi^{\tau}_l)\bigg]\mathbf
1\big[\chi(t)=\chi(l,\tau)\big],\label{eqn:update-v}
\end{align}
where $c(l,\tau,t)=\sum_{t'=0}^t \mathbf I [\chi(t')=\chi(l,\tau)]$
is the number of updates of the representative state $\chi(l,\tau)$
up to slot $t$, $\chi^I\triangleq\{\chi^1_1, \cdots, \chi^1_L\}$
denotes the reference state and $\bar t\triangleq \sup
\{t|\chi(t)=\chi^I\}$\footnote{From \eqref{eqn:update-v}, we can
observe that
$\widetilde{V}_{l,t}(\chi^1_l)=\widetilde{V}_{l,0}(\chi^1_l)=0$
($\forall t>0$).}.

\item Step 4 (\textbf{Termination}): If
$||\widetilde{\mathbf{V}}_{t}-\widetilde{\mathbf{V}}_{t-1}||<\delta_v$,
stop; otherwise, set $t:=t+1$ and go to Step 2.
\end{itemize}~ \hfill\QEDclosed\label{alg:app-v}
\end{Alg}

Using the theory of stochastic approximation on the update equation
in Step 3, the convergence of the above online algorithm is given
below:
\begin{Lem}[Convergence of Algorithm \ref{alg:app-v}] \label{lem:conv}
Denote
\begin{equation}
\mathbf{A}_{t-1} = (1-\epsilon_{t-1}) \mathbf{I} + \mathbf{M}^{-1}
\mathbf{F}(\Omega_t) \mathbf{M} \epsilon_{t-1} \quad \mbox{and}
\quad \mathbf{B}_{t-1} = (1-\epsilon_{t-1}) \mathbf{I} +
\mathbf{M}^{-1} \mathbf{F}(\Omega_{t-1}) \mathbf{M} \epsilon_{t-1},
\label{eqn:ch5-cov-condition}
\end{equation}
where $\Omega_t$ is the unichain control policy at slot $t$,
$\mathbf{F}(\Omega_t)$ is the transition matrix
under the unichain system control policy $\Omega_t$,
and $\mathbf{I}$ is the identity
matrix. If for the entire sequence of control policies
$\{\Omega_t\}$, there exists an $\delta_t>0$ and some positive
integer $\beta$ such that
\begin{equation}
[\mathbf{A}_{\beta-1}\cdots\mathbf{A}_1]_{(i,I)} \geq \delta_t,
\quad [\mathbf{B}_{\beta-1}\cdots\mathbf{B}_1]_{(i,I)} \geq \delta_t
\quad \forall i, \nonumber
\end{equation}
where $[\cdot]_{(i,I)}$ denotes the element in the
$i$-th row and the $I$-th column and
$\delta_t=\mathcal{O}(\epsilon_t)$, then the
following statements are true:
\begin{itemize}
\item The update of the parameter vector will converge almost surely for any given initial parameter vector
$\widetilde{\mathbf{V}}_0$, i.e.\textcolor{red}{,}
$\lim\limits_{t\rightarrow +\infty} \widetilde{\mathbf{V}}_t =
\widetilde{\mathbf{V}}_{\infty}$ a.s..
\item The steady state parameter vector
$\widetilde{\mathbf{V}}_{\infty}$ satisfies: $\theta \mathbf{e} +
\widetilde{\mathbf{V}}_{\infty} = \mathbf{M}^{-1} \mathbf T (
\mathbf{M} \widetilde{\mathbf{V}}_{\infty} )$, where $\theta$ is a
constant.
\end{itemize}~\hfill\QEDclosed
\label{lem:app-v}
\end{Lem}

\begin{proof}
Please refer to Appendix A.
\end{proof}

\begin{Rem}[Interpretation of the Conditions in Lemma \ref{lem:app-v}]
Note that $A_t$ and $B_t$ are related to an equivalent transition
matrix of the underlying Markov chain. Eqn.
\eqref{eqn:ch5-cov-condition} simply means that the system state
$\chi^I$ is accessible from all the system states after some finite
number of transition steps. This is a very mild condition and
is satisfied in most of the cases we are
interested.~ \hfill\QEDclosed
\end{Rem}

\begin{Exam}[Approximating Potential Functions for  Uplink OFDMA Systems]
In the example of the uplink OFDMA system example, we consider
packet flows and assume Poisson packet arrival with average arrival
rate $\overline{\lambda}_l$ (packet/s) and exponential packet size
distribution with mean packet size $\overline{N}_l$ (bit/packet) for
the $l$-th MS. Given a stationary policy, define
the conditional mean departure rate of packets of link $l$
(conditioned on the system state $\chi$) as
$\overline{\mu}_l(\chi)=\mu_l(\chi)/\overline{N}_l=\sum_{m=1}^{N_F}
R_{l,m}(\chi)/\overline{N}_l$ (packet/s). Moreover, we assume that
the scheduling slot duration (or frame duration) $\tau$ (s/slot) is
substantially smaller than the average packet inter-arrival time as
well as the average packet service time
($\overline{\lambda}_l\tau \ll 1$ and $\overline{\mu}_l(\chi)\tau
\ll 1$)\footnote{This assumption is reasonable in practical systems.
For instance, in the UL WiMAX (with multiple UL users served
simultaneously), the minimum resource block that could be allocated
to a user in the UL is $8 \times 16$ symbols $-$ 12 pilot
symbols$=$116 symbols. Even with 64QAM and rate $\frac{1}{2}$
coding, the number of payload bits it can carry is
$116\times3$bits$=$348 bits. As a result, when there are a lot of UL
users sharing the WiMAX AP, there could be cases that the MPEG4
packet (around 10K bits) from an UL user cannot be delivered in one
frame. In addition, the delay requirement of MPEG4 is 500ms or more,
while the frame duration of Wimax is 5ms. Hence, it is not necessary
to serve one packet during one scheduling slot so that the scheduler
has more flexibility in allocating resources.
Therefore, in practical systems, an application level packet may
have mean packet length spanning over many time slots (frames) and
this assumption is also adopted in
\cite{Sadiq:2009,Baris:2009,Jennifer:2009,Crabill:1972}.}. There
is a packet departure from the $l$-th queue at the
$(t + 1)$-th slot if the remaining service time of a packet is less
than the current slot duration $\tau$. By the memoryless property of
the exponential distribution, the remaining packet length (also
denoted as $N_l(t)$) at any slot t is also exponentially distributed.
Thus, the conditional probability of a packet departure event at the
$t$-th slot is given by
\begin{align}
&\Pr
\left[\frac{N_l(t)}{\mu_l(t)}<\tau|\chi_l(t),\Omega_l(\chi(t))\right]=\Pr\left[\frac{N_l(t)}{\overline
N_l}<\overline{\mu}_l\big(\chi(t)\big)\tau\right]\nonumber\\
=&1-\exp
\Big(-\overline{\mu}_l\big(\chi(t)\big)\tau\Big)\approx\overline{\mu}_l\big(\chi(t)\big)\tau.
\end{align}
Note that under assumption $\overline{\lambda}_l\tau \ll 1$ and
$\overline{\mu}_l(\chi)\tau \ll 1$, the probability for simultaneous
arrival, departure of two or more packets from the same queue or
different queues and simultaneous arrival as well as departure in a
slot are $\mathcal O\big((\overline \lambda_l \tau)^2\big),\
\mathcal O\big((\overline \mu_l(\chi) \tau)^2\big)$ and $\mathcal
O\big((\overline \lambda_l \tau)\cdot(\overline \mu_l(\chi)
\tau)\big)$ respectively, which are asymptotically negligible.
Hence, the queue dynamics of each link becomes a  controlled
birth-death process with the transition probability of each link
given by
\begin{align}
&\Pr \left[ \chi_l(t+1)=(\mathbf{H}_l(t+1),Q_l(t+1)) \big|
\chi_l(t)=(\mathbf{H}_l(t),Q_l(t)), \Omega_l(\chi(t))
 \right]\nonumber\\
  = &\left\{
\begin{array}{ll}
\Pr [ \mathbf{H}_l(t+1) | \mathbf{H}_l(t) ] \overline{\lambda}_l\tau, & Q_l(t+1)=Q_l(t)+1\\
\Pr [ \mathbf{H}_l(t+1) | \mathbf{H}_l(t) ] \overline{\mu}_l(\chi(t))\tau, & Q_l(t+1)=Q_l(t)-1 \\
\Pr [ \mathbf{H}_l(t+1) | \mathbf{H}_l(t) ]
\big(1-\overline{\mu}_l(\chi(t))\tau-\overline{\lambda}_l\tau\big),
& Q_l(t+1)=Q_l(t)
\end{array}.
\right.
\end{align}
With the above assumptions, the optimization problem in
\eqref{eqn:app_Q_con} can be transformed into
\begin{align}
\min_{\Omega = \{ \Omega_{l,s}, \Omega_{l,p} \}}\ & \sum_{l=1}^L
\left( \sum_{m=1}^{N_F} \gamma_l p_{l,m} -
\frac{\tau}{\overline{N}_l}\Delta \widetilde{V}(Q_l)
\left(\sum_{m=1}^{N_F}s_{l,m}\log(1+p_{l,m}|H_{l,m}|^2)\right)\right),\
\forall \chi=(\mathbf H, \mathbf Q)
\label{eqn:prob4aa}\\
\text{s.t.}\ & \eqref{eqn:ch3_subc-constraint} \quad \text{is satisfied},\nonumber
\end{align}
where $$\Delta
\widetilde{V}(Q_l)=\mathbf{E}_{\mathbf{H}^j_l}\left[\widetilde{V}_l(\mathbf{H}^j_l,Q_l)|\mathbf{H}_l
\right]-\mathbf{E}_{\mathbf{H}^j_l}\left[\widetilde{V}_l(\mathbf{H}^j_l,Q_l-1)|\mathbf{H}_l
\right].$$ Using  standard optimization techniques,
the subcarrier and power allocation is given by
\begin{eqnarray}
p_{l,m}(\mathbf{H}, \mathbf{Q})&=&s_{l,m}(\mathbf{H},
\mathbf{Q})\left(\frac{\frac{\tau}{\overline{N}_l}
\Delta \widetilde{V}(Q_l)}{\gamma_l}-\frac{1}{|H_{l,m}|^2}\right)^{+}\label{eqn:ch_5_joint-pwr-allo}\\
s_{l,m}(\mathbf{H}, \mathbf{Q}) &=& \left\{
\begin{array}{ll} 1, &
\text{if} \quad X_{l,m}=\max_{j} \big\{X_{j,m}\big\} >0\\
0, & \textrm{otherwise}
\end{array} \right. \label{eqn:ch_5_joint-subc-allo}
\end{eqnarray}
where $$X_{l,m}=\frac{\tau}{\overline{N}_l}\Delta \widetilde{V}(Q_l)
\log\Big(1+|H_{l,m}|^2\big(\frac{\frac{\tau}{\overline{N}_l}\Delta
\widetilde{V}(Q_l)}{\gamma_l}-\frac{1}{|H_{l,m}|^2}\big)^{+}\Big)-\gamma_l
\big(\frac{\frac{\tau}{\overline{N}_l}\Delta
\widetilde{V}(Q_l)}{\gamma_l}-\frac{1}{|H_{l,m}|^2}\big)^{+}.$$
Hence, the control action calculation (Step 2 of Algorithm
\ref{alg:app-v}) and per-link potential update (Step 3 of Algorithm
\ref{alg:app-v}) are given below:
\begin{itemize}
\item {\bf Control Action Calculation}: Based on the realtime observation of the system state $\chi(t)=(\mathbf H (t), \mathbf Q(t))$,
perform subcarrier and power allocation according to
\eqref{eqn:ch_5_joint-pwr-allo} and \eqref{eqn:ch_5_joint-subc-allo}
at the $t$-th slot. In distributed implementation, each user
maintains its own per-link potential, i.e.\textcolor{red}{,} the
$l$-th user maintains the potential
$\{\widetilde{V}_l(\chi^{\tau}_l)|1\leq \tau \leq |\mathcal X_l|\}$.
According to (\ref{eqn:ch_5_joint-subc-allo}), the subcarrier
allocation can be determined distributively by an auction mechanism.
For example, each user $l$ submits a bid $X_{l,m}$ on each subcarrier $m$, and the user with the largest bid will get the
subcarrier. When the subcarrier allocation is determined, the power
allocation for each link can be calculated locally at each user
according to \eqref{eqn:ch_5_joint-pwr-allo}.
\item {\bf Per-link Potential Update:} Suppose in the $t$-th time slot the system is in the reference
state $\chi(l,\tau)$, i.e.\textcolor{red}{,} $\chi(t)=\chi(l,\tau)$,
the $l$-th user will update the per-link potential
$\widetilde{V}_l(\chi^{\tau}_l)$ according to
\begin{align}
&\widetilde{V}_{l,t+1}(\chi^{\tau}_l)=\widetilde{V}_{l,t}(\chi^{\tau}_l)+\epsilon_{c(l,\tau,t)}\Big[
\Big(\nu_l Q_l(t)+ \gamma_l \sum_{m=1}^{N_F}
p_{l,m}(t)+\eta_l \mathbf{I} [Q_l(t)=N_Q]\nonumber\\
&+\sum_l \mathbf E_{\mathbf H_l'}[\widetilde{V}_{l,t}(\mathbf H_l',
Q_l(t+1))|\mathbf H_l(t)]\Big)-
\sum_l \mathbf E_{\mathbf H_l'}[\widetilde{V}_{l,t}(\mathbf
H_l',Q_l(\bar t+1))|\mathbf H_l(\bar
t)]-\widetilde{V}_{l,t}(\chi^{\tau}_l)\Big].\label{eqn:update-v-example}
\end{align}
\end{itemize}
\begin{Rem}[Implementation Considerations]
Note that we choose the reference state as
$\chi^I\triangleq\{\chi^1_1,\chi^1_2, \cdots, \chi^1_L\}$ with
$\chi^1_l=(\mathbf H_l^1, Q_l^1)$, where $\mathbf H_l^1$ can be any
fixed local CSI while the local QSI $Q_l^1=0$
(i.e., the buffer is empty). Hence, each source
node of the $l$-th link requires only the local CSI $\mathbf H_l$,
the local QSI $\mathbf Q_l$ as well as some potential functions of
the other links $\left\{\mathbf{E} \left[
\widetilde{V}_{l'}(\mathbf{H}^j_{l'},Q_{l'}) | \mathbf{H}^1_{l'}
\right]\Big| Q_{l'}=0,1\right\}$ in order to compute the update in
\eqref{eqn:update-v-example}. While the computational complexity and
signaling overhead have been substantially reduced compared with the
brute-force centralized solution, the computation and the overhead
of delivering the terms $\left\{\mathbf{E} \left[
\widetilde{V}_{l'}(\mathbf{H}^j_{l'},Q_{l'}) | \mathbf{H}^1_{l'}
\right]\Big| Q_{l'}=0,1\right\}$ to all the nodes are still quite heavy.
In the next section, we  elaborate on a second
approximation approach which could further simplify the complexity
and overhead. ~ \hfill\QEDclosed
\end{Rem}
 \label{exam:app-v}
\end{Exam}

\subsection{Approach 2: Approximating Q-Factors} \label{sec_MDP_sub_appro_q}

In this section, we  propose another approach to address the
complexity issue and signaling overhead issue by approximating Q-factors. Different from the approach of approximating the potential
function, this approach could establish a totally distributed
learning algorithm at each node of the system.  From
the Bellman equation of the delay-optimal MDP in
\eqref{eqn:ch_5_bellman}, the Q-factor is defined as
\begin{equation}
\mathcal{Q}(\chi,a)\triangleq g(\chi,a)+
\sum_{\chi'}\Pr[\chi'|\chi,a]V(\chi')-\theta, \ \forall \chi,
\label{eqn:ch_5_q}
\end{equation}
where $a$ is an arbitrary action in the action space
$\mathcal{A}$. Hence, we have
$$V(\chi) =\min_{a} \mathcal{Q}(\chi,a), \ \forall \chi,$$ and $\mathcal{Q}(\chi,a)$ satisfies the following ``Q-factor form'' of the Bellman equation
\begin{equation}\label{eqn:ch5_bellman_q_factor}
\mathcal{Q}(\chi,a)
=g(\chi,a)+\sum_{\chi'}\Pr[\chi'|\chi,a]\min_b
\mathcal{Q}(\chi',b)-\theta, \ \forall \chi.
\end{equation}
Moreover, the optimal control policy is given by
\begin{equation}
\Omega^{*}(\chi) = \min_{a \in \mathcal{A}} \mathcal{Q} (\chi,a), \
\forall \chi.\label{eqn:app-q-plc}
\end{equation}

As an illustration, we consider the uplink OFDMA
system  example similar to Section \ref{sec_MDP_sub_appro}.
Similar to Section \ref{sec_MDP_sub_appro},
we  approximate the Q-factor in (\ref{eqn:ch_5_q}) by a linear
approximation given by:
\begin{equation}
\mathcal{Q}(\chi,a) \approx \sum_{l=1}^L  q_l(\chi_l,a_l), \ \forall
\chi, \label{eqn:app-q}
\end{equation}
where $a_l$ denotes the local actions (such as the local subcarrier
allocation, local power allocation, precoder design, etc) of the
$l$-th link (thus $a=\{a_l \}$), $q_l(\chi_l,a_l)$ is
referred to as the {\em per-link Q-factor} for the
$l$-th link of local system state $\chi_l$ and action $a_l$.
Moreover, the per-link Q-factor is defined as the solution of the
following per-link fixed-point equation:
\begin{equation}
q_l(\chi_l,a_l) = g_l ( \chi_l,a_l ) + \sum_{\chi_l'\in
\mathcal{X}_l} \Pr [\chi_l'|\chi_l,a_l]
 W_l(\chi_l') - \theta_l, \label{eqn:ch_5_q_l}
\end{equation}
where
\begin{equation}
W_l(\chi_l) = \mathbf{E}_{\mathbf{H}_l} \left[
 \min_{a_l\backslash \mathbf{s}_l} \left[ q_l \left(\chi_l,\{s_{l,m}=\mathbf{I}(|H_{l,m}|\geq H^*_{L-1})\},a_l \right) \right] \right], \nonumber
\end{equation}
\begin{equation}
g_l (\chi_l,a_l) = \nu_l Q_l + \gamma_l \sum_{m=1}^{N_F}
p_{l,m}+\eta_l \mathbf I[Q_l=N_Q] \nonumber,
\end{equation}
$\mathbf{s}_l = \big\{s_{l,m}| \forall m\in \{1,2,\cdots, N_F\}\big\}$,
$a_l\backslash\mathbf{s}_l$ denotes all the control action except
the subcarrier selection $\mathbf{s}_l$, and
$H_{L-1}^*$ denotes
the largest order statistic of the $L-1$ i.i.d. random variables,
each of which has the same distribution as the
channel fading. Therefore, an online Q-Learning algorithm for estimating
per-link Q-factors is given below:
\begin{Alg}[Online Algorithm for Estimating Per-Link Q-factors]
$\quad$
\begin{itemize}
\item Step 1 ({\bf Initialization}): Start with an initial per-link Q-factor $\{q_{l,0}(\chi_l,a_l)\}$ with $q_{l,0}(\chi_l^I,a_l^I)=0$ ($\forall l \in \mathcal L$).

\item Step 2 ({\bf Calculate Control Actions}): Based on the realtime observation of the system state $\chi(t)$ at slot $t$, calculate control
actions according to:
\begin{equation}
a^*_t = \{a_{1,t}^*,a_{2,t}^*,\cdots,a_{L,t}^*\}= \arg\min_{a}
\mathcal{Q}_t(\chi(t),a) = \arg\min_{\{a_{l}\}} \sum_{l =1}^L
q_{l,t} (\chi_l(t),a_l)\nonumber.
\end{equation}

\item Step 3 ({\bf Update  Per-Link Q-Factors}):
After the control action has been determined,
update all the per-link potentials $\{q_{l}(\chi_l,a_l)\}$ ($\forall l \in \mathcal L$) based on the
real-time observations of the local per-link system state
$\chi_l(t)$ ($\forall l \in \mathcal L$), where $\chi_l(t)=(\mathbf
H_l(t),Q_l(t) ) $ as follows:
\begin{align}
q_{l,t+1}(\chi_l^{\tau},a_l) = &q_{l,t}(\chi_l^{\tau},a_l) +
\epsilon_{c_l(\tau,a_l,t)} \Big[  \big(g_l ( \chi_l^{\tau},a_l ) +
W_{l,t}(\chi_l(t+1))\big)
\nonumber\\
&- \big(g_l(\chi_l^I,a^I_{l})+ W_{l,t}(\chi_l(\bar
t_l+1)\big) - q_{l,t}(\chi_l^{\tau},a_l)\Big]\mathbf I[(\chi_l^{\tau},a_l)=(\chi_l(t),a^*_{l,t})],
\label{eqn:update-q}
\end{align}
where $c_l(\tau,a_l,t)=\sum_{t'=0}^t \mathbf
I[(\chi_l^{\tau},a_l)=(\chi_l(t'),a^*_{l,t'})]$ is the number of
updates of the state-action pair $(\chi_l^{\tau},a_l)$ up to slot
$t$, $(\chi^I_l,a^I_l)$ denotes the reference state-action pair\footnote{From
\eqref{eqn:update-q}, we can observe that
$q_{l,t}(\chi_l^I,a_l^I)=q_{l,0}(\chi_l^I,a_l^I)=0$ ($\forall
t>0$).}
of the $l$-th link and $\bar t_l\triangleq \sup
\{t|(\chi_l(t),a_{l,t}^*)=(\chi_l^I,a^I_{l})\}$.

\item Step 4 ({\bf Termination}):
If $\sum_{l \in \mathcal{L}}||\mathbf q_{l,t}-\mathbf
q_{l,t-1}||<\delta_q$, stop; otherwise, set $t:=t+1$ and go to Step
2.
\end{itemize}~ \hfill\QEDclosed\label{alg:app-q}
\end{Alg}

\begin{Rem}[Implementation Considerations]
Using the linear approximation in (\ref{eqn:app-q}), the dimension
of the Q-factor (and hence the computational
complexity) is significantly reduced. Furthermore, the online update
procedure in step 3 can be implemented locally at each node,
requiring only knowledge of  the local CSI $\mathbf
H_l$ and the local QSI $Q_l$.~\hfill\QEDclosed
\end{Rem}

Similarly, using the theory of stochastic approximation on the
update equation (\ref{eqn:update-q}) in Step 3, we summarize the
convergence of the above online learning algorithm as follows:
\begin{Lem}[Convergence of Algorithm \ref{alg:app-q}]
Denote $\mathbf{q}_l = \{q_l(\chi_l,a_l)\}$ and
\begin{equation}
\mathbf{A}_{t-1} = (1-\epsilon_{t-1}) \mathbf{I} +
\mathbf{F}(\Omega_t) \epsilon_{t-1} \quad \quad \mathbf{B}_{t-1} =
(1-\epsilon_{t-1}) \mathbf{I} + \mathbf{F}(\Omega_{t-1})
\epsilon_{t-1} ,\nonumber
\end{equation}
where $\Omega_t$ is the unichain system control policy at the $t$-th
frame, $\mathbf{F}(\Omega_t)$ is the transition matrix  under the unichain system control policy $\Omega_t$,
and $\mathbf{I}$ is the identity
matrix. If for the entire sequence of control policies
$\{\Omega_t\}$, there exists an $\delta_t>0$ and some positive
integer $\beta$ such that
\begin{equation}
[\mathbf{A}_{\beta-1}\cdots\mathbf{A}_1]_{(i,I)} \geq \delta_t,
\quad [\mathbf{B}_{\beta-1}\cdots\mathbf{B}_1]_{(i,I)} \geq \delta_t
\quad \forall a, \nonumber
\end{equation}
where $[\cdot]_{(i,I)}$ denotes the element in the $i$-th row and the $I$-th
column and $\delta_t=\mathcal{O}(\epsilon_t)$, then
the following statements are true:
\begin{itemize}
\item The update of the per-link Q-factor will converge almost surely for any given initial per-link Q-factor $\{\mathbf{q}_{l}^{0}\}$, i.e., $\lim\limits_{t\rightarrow +\infty} \mathbf{q}_{l,t} =
\mathbf{q}_{l,\infty}$ a.s. ($\forall l\in \mathcal{L}$).
\item The steady state per-link Q-factor
$\mathbf{q}_l^{\infty}$ satisfies:
\begin{equation}
q_{l,\infty}(\chi_l,a_l) = g_l ( \chi_l,a_l ) + \sum_{\chi_l'\in
\mathcal{X}_l} \Pr [\chi_l'|\chi_l,a_l]
 W_{l,\infty}(\chi_l') - \theta_l, \nonumber
\end{equation}
where $\theta_l$ is a constant and
\begin{equation}
W_{l,\infty}(\chi_l) = \mathbf{E}_{\mathbf{H}_l} \left\{
 \min_{a_l\backslash\mathbf{s}_l} \left[ q_{l,\infty} \left(\chi_l,\{s_{l,m}=\mathbf{I}(H_{l,m}\geq H^*_{L-1})\},a_l \right) \right] \right\}. \nonumber
\end{equation}
\end{itemize}~\hfill\QEDclosed
 \label{lem:app-q}
\end{Lem}

The proof of Lemma \ref{lem:app-q} follows a similar approach as the
proof of Lemma \ref{lem:app-v}. In the following, we elaborate
Algorithm \ref{alg:app-q} using the uplink OFDMA
system example.

\begin{Exam}[Approximating Q-Factors for Uplink OFDMA Systems]
Consider the uplink OFDMA system example under the same
assumptions as in Example \ref{exam:app-v}. Since the power control
can be determined locally given a subcarrier allocation action, the
per-link Q-factor is defined as $\{q_l(\chi_l,\mathbf{s}_l)|\forall
\chi_l,\mathbf s_l\}$ ($\forall l\in \mathcal L$), satisfying
\begin{align}
q_l(\chi_l,\mathbf{s}_l) =& \min_{\{p_{l,m}\}} \left\{ g_l (
\chi_l,\mathbf{s}_l ) + \sum_{\chi_l' \in \mathcal{X}_l} \Pr
[\chi_l'|\chi_l,\mathbf{s}_l]
 W_l(\chi_l') - \theta_l \right\}\label{eqn:q-d}\\
 =& \min_{\{p_{l,m}\}} \Bigg\{ \nu_l Q_l + \sum_{m=1}^{N_F} \gamma_l p_{l,m}+\eta_l \mathbf I[q_l=N_Q]+\overline W_l(\mathbf H_l, Q_l)+\Delta \overline W_l (\mathbf H_l, Q_l+1)
] \nonumber\\
&\quad \quad - \frac{\tau}{\overline{N}_l} \Delta \overline W_l (\chi_l) \left( \sum_{m=1}^{N_F} s_{l,m} \log(1+p_{l,m}|H_{l,m}|^2 )\right)-\theta_l\Bigg\} \nonumber,
\end{align}
where $\overline W(\mathbf{H}_l,Q_l)=\mathbf{E}_{\mathbf{H}'_l}\big[
W_l(\mathbf{H}'_l,Q_l)|(\mathbf{H}_l,Q_l) \big]$ and $\Delta
\overline W(\mathbf H_l, Q_l)=\overline W(\mathbf{H}_l,Q_l)-\overline
W(\mathbf{H}_l,Q_l-1)$. Due to the symmetry of each subcarrier and
the birth-death queue dynamics, the per-link Q-factor satisfying
\eqref{eqn:q-d} can be written as the summation of the per-link per-subcarrier
Q-factors
\begin{eqnarray}
q_l(\chi_l,\mathbf{s}_l) = \sum_{m=1}^{N_F}
\widetilde{q}_l(H_{l,m},Q_l,s_{l,m}) \nonumber,
\end{eqnarray}
where the per-link per-subcarrier Q-factor $\{\widetilde{q}_l(H,Q,s)\}$
satisfies the following per-link per-subcarrier fixed-point equation:
\begin{align}
&\widetilde{q}_l(H_{l,m},Q_l,s_{l,m})\nonumber\\
=& \min_{\{p_{l,m}\}} \bigg\{ \nu_l \frac{Q_l}{N_F } + \gamma_l p_{l,m}+\frac{\eta_l \mathbf I[Q_l=N_Q]}{N_F}  + \frac{\sum_{m=1}^{N_F}\big( \overline \nu_l (H_{l,m},Q_l)+ \Delta \overline \nu_l (H_{l,m}, Q_l+1)\big)}{N_F} \nonumber\\
 & \quad \quad \quad - \frac{\tau}{\overline{N}_l}\big (\sum_{m=1}^{N_F}\Delta \overline \nu_l(H_{l,m},Q_l)\big) s_{l,m} \log(1+p_{l,m}|H_{l,m}|^2 ) -
 \frac{\theta_l}{N_F}
 \bigg\} \nonumber,
\end{align}
where $\nu_l (H_{l,m},Q_l)=\mathbf E \big[ \widetilde
q_l (H_{l,m},Q_l,s_{l,m}=\mathbf{1}[|H_{l,m}|\geq H^*_{L-1}]) |(
H_{l,m},Q_l) \big]$, $\overline \nu_l (H_{l,m},Q_l)=\mathbf
E_{H_{l,m}'}[\nu_l (H_{l,m}',Q_l)|(H_{l,m},Q_l)]$ and $\Delta
\overline \nu_l (H_{l,m},Q_l)=\overline \nu_l
(H_{l,m},Q_l)-\overline \nu_l (H_{l,m},Q_l-1)$.  According to
\eqref{eqn:app-q-plc}, the subcarrier allocation is given by
\begin{equation}
s_{l,m} = \left\{
\begin{array}{ll} 1, &
\text{if} \quad \widetilde{q}_l(H_{l,m},Q_l, s_{l,m}=1)=\min_{j} \widetilde{q}_j(H_{j,m},Q_l,s_{j,m}=1) \\
0, & \textrm{otherwise}
\end{array} \right..\label{eqn:ch5_per-subc-q-subc-alloc}
\end{equation}
Moreover, given the subcarrier allocation, by the
optimization techniques, the power allocation is given by
\begin{equation}
p_{l,m}(\mathbf{H}, \mathbf{Q}) = s_{l,m}
\left(\frac{\frac{\tau}{\overline{N}_l}\sum_{m=1}^{N_F}\Delta
\overline \nu_l
(H_{l,m},Q_l)}{\gamma_l}-\frac{1}{|H_{l,m}|^2}\right)^{+}.
\label{eqn:ch5_per-subc-q-pwr-alloc}
\end{equation}
Hence, the control action calculation (Step 2 of Algorithm
\ref{alg:app-q}) and Q-factor update (Step 3 of Algorithm
\ref{alg:app-q}) are given below:
\begin{itemize}
\item {\bf Control Action Calculation}:
Based on the realtime observation of the system state
$\chi(t)=(\mathbf H (t), \mathbf Q(t))$ at the $t$-th slot, we
determine the subcarrier allocation distributively by an auction
mechanism according to \eqref{eqn:ch5_per-subc-q-subc-alloc}: each
user $l$ submits one bid
$\widetilde{q}_{l,t}(H_{l,m},Q_l,s_{l,m}=1)$ on each subcarrier. The
user with the minimum bid will get the subcarrier. Given the
subcarrier allocation, the power allocation can be calculated
locally at each user $l$ according to
\eqref{eqn:ch5_per-subc-q-pwr-alloc}.
\item {\bf Q-Factor Update}:
After the control action is determined, update  all the per-link
per-subcarrier potentials $\{\widetilde q_{l}(H,Q,s)\}$
($\forall l \in \mathcal L$) based on the real-time observations of
the local per-link system state $\chi_l(t)=(\mathbf H_l(t),Q_l(t) ) $ as follows:
\begin{align}
\widetilde q_{l,t+1}(H,Q,s) = & \widetilde q_{l,t}(H,Q,s) +
\epsilon_{c_l(Q,H,s,t)} \Big[ \big(\nu_l \frac{Q}{N_F } + \gamma_l
p_{l,m}(t) + \frac{\eta_l \mathbf{I}[Q_l=N_Q]}{N_F} + \overline
w_{l,t}(Q_l(t+1))\big)
\nonumber\\
& -  \overline w_{l,t}(Q_l(\bar t_l+1)- q_{l,t}(H,Q,s)\Big]\mathbf
I[\cup_{m=1}^{N_F}\{(H,Q,s)=\big(H_{l,m}(t),Q_l(t),s_{l,m}(t)\big)\}],
\label{eqn:update-example-q}
\end{align}
where $c_l(H,Q,s,t)=\sum_{t'=0}^t \mathbf
I[\cup_m\{(H,Q,s)=(H_{l,m}(t'),Q_l(t'),s_{l,m}(t'))\}]$ is the
number of updates of the per-subcarrier state-action pair $(H,Q,s)$ up to
slot $t$\textcolor{red}{,}  $(H^I,Q^I,s^I)$ denotes the reference
state-action pair of each link, and $\bar t_l\triangleq \sup
\{t|\cup_m\{(H^I,Q,^I,s^I)=(H_{l,m}(t'),Q_l(t'),s_{l,m}(t'))\}$.

Notice that the above Q-factor update requires only the local
information at each user, and it does not
lead to any signaling overhead.
\end{itemize}~\hfill\QEDclosed
\end{Exam}

\begin{Rem}[Comparison of the Two Approximate MDP Approaches]
Both of the approximate MDP
approaches in Section \ref{sec_MDP_sub_appro} and
Section \ref{sec_MDP_sub_appro_q} can effectively reduce the
complexity and signaling overhead in the MDP
solutions. However, there are pros and cons in the
two approaches:
\begin{itemize}
\item In general, approximating potential functions will lead to
fewer dimensions (and lower memory
requirement) than approximating Q-factors. This is because
a Q-factor depends on both the
system state and the control action.
\item In the distributed implementation of the online learning
algorithm, updates of the per-link
Q-factor can be done locally without any signaling overhead, whereas
updates of the per-link potential function
still require some information exchange
among the nodes.
\item In some cases, computing actions from
potential functions may still be complicated compared with
computing actions from  Q-factors.
\end{itemize}
Although approximate MDP and stochastic learning
can effectively reduce the complexity in the MDP solution, extension
to multi-hop networks is far from trivial due to the complex
interactions of the queue dynamics and the huge state space
involved. More investigations are needed regarding how to
approximate  potential functions or Q-factors as well as the associated
convergence proof in multi-hop networks.~\hfill\QEDclosed
\end{Rem}

\section{Delay-Aware Routing in Multi-hop Wireless Networks}
\label{sec_routing}

In this section, we   focus on delay-aware routing
in wireless multi-hop networks  using the Lyapunov stability drift
approach. Due to the complex coupled queue dynamics in multi-hop
wireless networks, extensions of the  equivalent rate constraint
approach and the   approximate MDP  approach to multi-hop networks
are highly non-trivial. On the other hand, the Lyapunov drift
approach can be easily applied in multi-hop networks to derive
dynamic control algorithms that are adaptive to both the system CSI
and the system QSI. Hence, the Lyapunov drift approach is receiving
more and more attention recently in multi-hop networks. In the
following, we  first review the
traditional DBP routing in wireless multi-hop networks and then
focus on various delay reduction techniques in the
enhanced DBP routing
\cite{Neely-Modiano-Rohrs:2005,BuiSrikant:2009,YingShakkottai08oncombSPDBP,NeelyLMDBP11,StandfordLIFO11,NeelyLIFO11,JavidiORCD1:2010,JavidiORCD2:2010,Neely09DIVBAR}.

\subsection {Traditional DBP Routing}

The traditional DBP routing in wireless multi-hop
networks was originally proposed in the seminal paper
\cite{Tassiulas-Ephremides:1992-2} to maximize the stability region
and then extended in
\cite{Neely-Modiano-Rohrs:2005,Georgiadis-Neely-Tassiulas:2006}.
This traditional DBP routing is illustrated below
\cite{Neely-Modiano-Rohrs:2005,Georgiadis-Neely-Tassiulas:2006}:
\begin{Alg} [Traditional DBP Routing]$\quad$
\begin{itemize}
\item {\bf Resource allocation}: For each commodity $c\in \mathcal C$ and
each link $l\in \mathcal L$, define the {\em backpressure of link
$l$ w.r.t. commodity $c$} as
\begin{align}
\Delta Q_l^{(c)}(t) \triangleq
Q_{s(l)}^{(c)}(t)-Q_{d(l)}^{(c)}(t).\label{eqn:tra-DBP-BP}
\end{align}
For each link $l\in \mathcal L$, define the {\em optimal commodity
of link $l$} as $c^*_l(t)\triangleq \arg \max_{c\in \mathcal C}\Delta
Q_l^{(c)}(t)$ and the {\em optimal backpressure of link $l$} as
$\Delta Q_l^*(t)\triangleq \max_{c\in \mathcal C}\{\Delta Q_l^{(c(t))}(t),0\}$.
Find the transmission rate $\boldsymbol \mu^*(t)$ such that
\begin{align}
\boldsymbol \mu^*(t)=\arg \max_{\boldsymbol \mu \in
\Lambda(t)}\sum_{l\in \mathcal L} \Delta Q_l^*(t)\mu_l.
\end{align}
\item {\bf Routing}: For each link $l$ such that $\Delta Q_l^*(t)>0$,
offer a transmission rate $\mu^*_l(t)$ to the data of commodity
$c^*_l(t)$ through link $l$.
\end{itemize} \label{Alg:tra-DBP-routing}
\end{Alg}

\begin{figure}
\centering
\includegraphics[height=5cm, width=10cm]{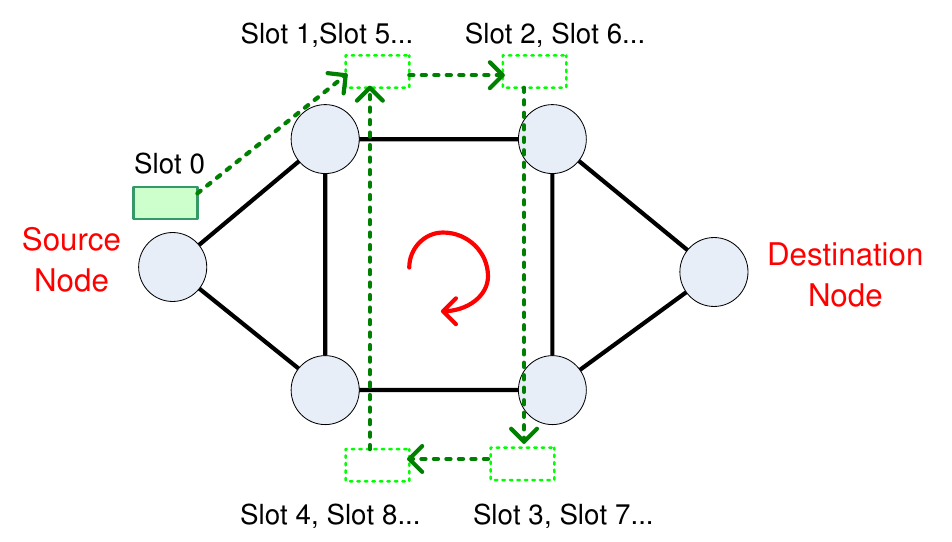}
\caption {Illustration of a single packet taking a periodic walk
under traditional DBP routing in a network.} \label{Fig:circulate}
\end{figure}

\begin{figure}
\centering
\includegraphics[height=3cm, width=10cm]{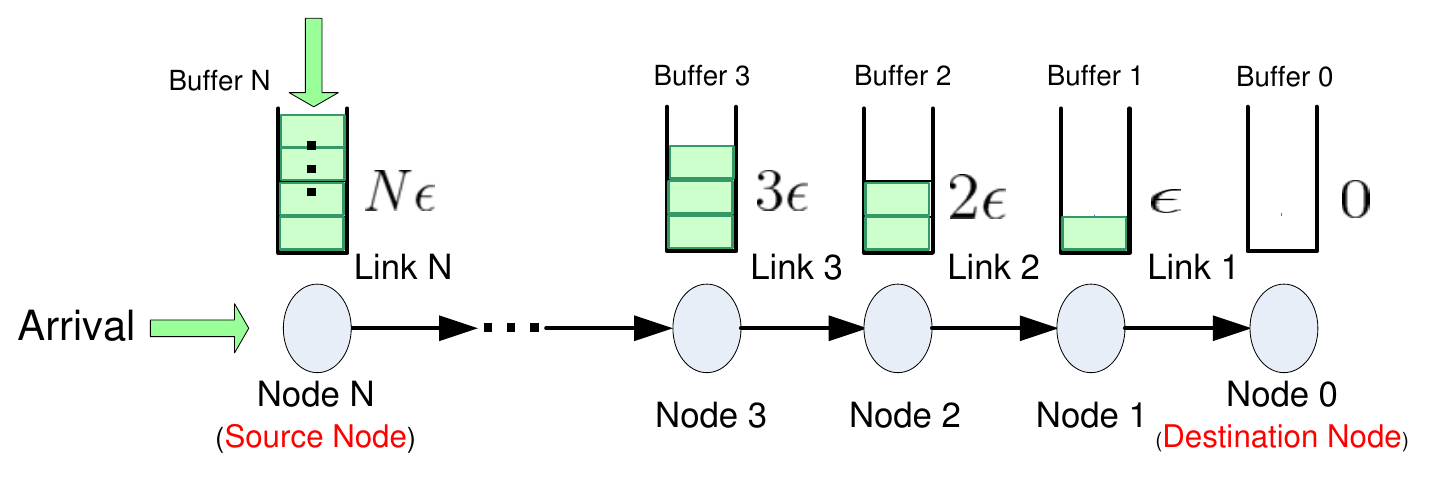}
\caption {Illustration of  differential backlogs under traditional
DBP routing in a tandem queueing network.} \label{Fig:tandem_queue}
\end{figure}

A significant weakness of the above DBP routing
algorithm is that it can suffer from very large delays due to the
following reasons.
\begin{itemize}
\item First, the traditional DBP routing exploits all possible paths
between source-destination pairs (i.e., load
balancing over the entire network) to maximize stability region
without considering the delay performance. This extensive
exploration is essential in order to maintain stability when the
network is heavily loaded. However, under light or moderate loads,
packets may be sent over unnecessarily long routes,
which leads to excessive delays. For example, if a single packet is
injected into an empty network, there is no backpressure to suggest
an appropriate path. Hence, the packet might take a random walk
through the network, or might take a periodic walk that never leads
to the destination, as illustrated in Fig. \ref{Fig:circulate}. In
this case, although the network congestion is quite low (only one
packet, i.e., zero average arrival rate), the
end-to-end delay can be infinity. Similarly, under light load, the
end-to-end delay can be large even though the average queue length
is guaranteed to be bounded by the Lyapunov drift theory
\cite{Neely-Modiano-Rohrs:2005,Georgiadis-Neely-Tassiulas:2006}.
Therefore, it is desirable to design a throughput-optimal routing
which exploits longer paths only when it is necessary.
\item Second, due to large queue sizes that must be maintained to
provide a gradient (backpressure) for each data flow, the DBP
routing can suffer from very large delays and the queues grow in
size with distance from the destination. To obtain some
insights on this, let us consider a flow traveling
through a $N$-hop tandem queueing network with $N+1$ nodes, as
illustrated in Fig. \ref{Fig:tandem_queue}. Let $Q_0$ be the queue
length at destination node 0 and $Q_n$ be the queue length of the
$n$-th upstream node from the destination node $0$, where
$n=1,2,\cdots,N$. Set $Q_0=0$. Under the traditional DBP routing algorithm, for a
link to be scheduled, the differential backlog associated with it
should be positive. Thus, $Q_1-Q_0=Q_1$ will be some positive
number, say $\epsilon$, and $Q_2-Q_1$ will be even larger than
$\epsilon$. For the purpose of obtaining some
insights, let us informally assume
$Q_2-Q_1=\epsilon$, which implies $Q_1=\epsilon$ and
$Q_2=2\epsilon$. Similarly, we can obtain $Q_n=n\epsilon$. Thus, the
total queue length for the flow under the traditional DBP routing
will be $\sum_{n=1}^N Q_n=\epsilon (1+2+\cdots+N)=\mathcal O(N^2)$
\cite{BuiSrikant:2009}. Therefore, it is desirable to design a
routing algorithm which can provide a sufficient gradient for each
data flow without causing too large delay for each packet.
\item Finally, the traditional DBP routing specifies a single next-hop
receiver before transmission, and hence does not exploit the
broadcast advantage of  multi-hop wireless networks when
wireless channels are unreliable (e.g., outage
probability without CSI). Due to multi-receiver diversity in
wireless channels, the probability of successful reception by at
least one node within a subset of potential receivers is much larger
than that of just one receiver. Therefore, it is desirable to design
flexible routing to dynamically adjust routing and scheduling
decisions in response to random outcome of each transmission.
\end{itemize}

Given the drawbacks of the DBP routing discussed
above, most recent studies try to improve the delay performance of
the DBP routing while maintaining its advantage in
throughput optimality. In the following, we  discuss three aspects
of delay reduction in DBP routing by utilizing the
shortest path concept
\cite{Neely-Modiano-Rohrs:2005,Georgiadis-Neely-Tassiulas:2006,BuiSrikant:2009,YingShakkottai08oncombSPDBP},
modifying the queueing disciplines
\cite{NeelyLMDBP11,StandfordLIFO11,NeelyLIFO11} and exploiting
receiver diversity over unreliable channels
\cite{JavidiORCD1:2010,JavidiORCD2:2010,Neely09DIVBAR} in wireless
multi-hop networks.

\subsection {Delay Reduction in DBP Routing by Shortest Path}

One of the major reasons for the poor end-to-end
delay performance of the traditional DBP routing algorithms is the
extensive exploration of routes. However, reducing
delay by restricting the routing constraint sets to some shorter
paths will reduce the stability region. In
\cite{Neely-Modiano-Rohrs:2005,Georgiadis-Neely-Tassiulas:2006,BuiSrikant:2009,YingShakkottai08oncombSPDBP},
the authors try to incorporate the idea of {\em shortest path
routing} into traditional DBP routing algorithms in different ways
while simultaneously maintaining throughput
optimality, which will be illustrated in the following.

The enhanced DBP routing algorithm proposed in
\cite{Neely-Modiano-Rohrs:2005} programs a {\em shortest path bias}
into the backpressure $\Delta Q_l^{(c)}(t)$ defined in
\eqref{eqn:tra-DBP-BP} so that in light or moderate loading
situations, nodes are inclined to route packets in the direction of
their destinations. Therefore, we call this enhanced DBP routing
algorithm  {\em shortest path bias DBP routing}. Specifically, the
{\em backpressure of link $l$ w.r.t. commodity $c$} in the enhanced
DBP routing algorithm is defined as
\begin{align}
\Delta Q_l^{(c)}(t) \triangleq \left(
Q_{s(l)}^{(c)}(t)+Z_{s(l)}^{(c)}\right)-\left(Q_{d(l)}^{(c)}(t)+Z_{d(l)}^{(c)}\right),\label{eqn:enh-DBP-BP-Neely}
\end{align}
where $Z_{n}^{(c)}$ is the shortest path bias at node $n$ for
commodity $c$. $Z_{n}^{(c)}$ can be chosen to be proportional to the
distance (or number of hops) between node $n$ to the destination of
commodity $c$ (where $Z_{n}^{(c)}=0$ if node $n$ is the destination
of commodity $c$). Besides the shortest path bias, the shortest path
bias DBP routing algorithm is the same as the traditional DBP
routing algorithm in Algorithm \ref{Alg:tra-DBP-routing}. It is
shown that the enhanced DBP algorithm through the shortest path bias
is still throughput-optimal. In addition, the
simulation results establish better delay performance of the
shortest path bias
 DBP routing than the traditional DBP routing
 \cite{Georgiadis-Neely-Tassiulas:2006}.

To reduce the end-to-end delay,
\cite{BuiSrikant:2009} introduces a cost function,
i.e., the total link rate in the network. Given a
set of packet arrival rates that lie within the stability region, the
total link rate can be used to measure the efficiency of the system
resource utilization. Thus, the min-resource routing problem is
formulated to find the routes, which minimizes the total link rate:
\begin{align}
\min_{\boldsymbol \mu\in\Lambda}\quad & \sum_{l\in \mathcal L}\sum_{c\in \mathcal C}
\mu_l^{(c)}\label{eqn:enh-DBP-min-res-prob}\\
s.t. \quad & \lambda^{(c)}_{n}+\sum_{l\in
\{l:d(l)=n\}}\mu^{(c)}_{l}\leq \sum_{l\in
\{l:s(l)=n\}}\mu^{(c)}_{l},\ \forall n \in \mathcal N, \ c\in
\mathcal C. \nonumber
\end{align}
Due to the nature of the cost function, shorter paths are preferred
over longer paths. For example, in a network with all links of equal
capacity, we prefer to have as few hops as possible to have good delay performance. The
associated routing algorithm  is called {\em min-resource DBP
routing} algorithm. Instead of \eqref{eqn:tra-DBP-BP}, it uses
\eqref{eqn:tra-DBP-BP} minus a parameter $V$ as the {\em
backpressure of link $l$ w.r.t. commodity $c$},
i.e.,
\begin{align}
\Delta Q_l^{(c)}(t) \triangleq
Q_{s(l)}^{(c)}(t)-Q_{d(l)}^{(c)}(t)-V.\label{eqn:enh-DBP-BP-Srikant}
\end{align}
Except for  parameter $V$, the min-resource DBP routing algorithm is
the same as the traditional DBP routing algorithm in Algorithm
\ref{Alg:tra-DBP-routing}. It is shown in \cite{BuiSrikant:2009}
that the average total link rate under the min-resource DBP routing
algorithm is within $\mathcal O(1/V)$ of the optimal value of the
optimization problem in \eqref{eqn:enh-DBP-min-res-prob}.
A larger $V$ corresponds to a
smaller delay and slower convergence speed to the stationary regime
while a smaller $V$ leads to a
larger delay and faster convergence speed. It is confirmed by
simulation results in \cite{BuiSrikant:2009} that the min-resource
DBP routing algorithm with a proper $V$ has better delay performance
than the traditional DBP routing algorithm.

The {\em joint traffic-splitting and
shortest-path-aided DBP routing} algorithm proposed in
\cite{YingShakkottai08oncombSPDBP} incorporates the shortest path concept
into the DBP routing by minimizing the average number of hops
between sources and destinations. Let $c\in \mathcal C$ denote a
flow (source-destination pair) in the multi-hop network, which is specified by its source and
destination, where $\mathcal C$ denotes the set of all the flows.
Let $A_{c,h}$ denote the fraction of flow $c$ transmitted over paths
with $h$ hops. Therefore, the average path-length minimization
problem is formulated as below\footnote{We omit the constraints of
the optimization problem in \eqref{eqn:enh-DBP-SP-prob} due to page
limit. Please refer to \cite{YingShakkottai08oncombSPDBP} for
details.}:
\begin{align}
\min_{\{A_{c,h}\}} \quad \sum_{c\in \mathcal C}\sum_{0\leq h \leq
N-1}V h A_{c,h}, \label{eqn:enh-DBP-SP-prob}
\end{align}
where $V$ is a positive constant and the optimal solution is the
same for all $V>0$. Note that $N-1$ is a universal upper bound on
the number of hops along loop-free paths. To realize
the shortest-path-aided routing, each node $n$
maintains a separate queue $(n,d(c),h)$ for the packets required to be
delivered to node $d(c)$ within $h\in \mathbb N$ hops and denote its
queue length at slot $t$ as $Q^{(c)}_{n,h}(t)$, where $\mathbb N$ denotes the set
of  natural numbers. Accordingly, define the {\em backpressure of
link $l$ w.r.t. flow $c$ and hop $h$} as follows:
\begin{align}
\Delta Q_{l,h}^{(c)}(t)\triangleq\begin{cases}
Q_{s(l),h}^{(c)}(t)-Q_{d(l),h-1}^{(c)}(t), & h-1\geq
H^{\min}_{d(l)\to c}\\
-\infty, & h-1< H^{\min}_{d(l)\to c}
\end{cases}\label{eqn:enh-DBP-BP-hopq-Shakkottai}
\end{align}
where $H^{\min}_{d(l)\to c}$ is the minimum number of hops, i.e.,
the length of the shortest path, required from node $d(l)$ to the
destination of flow $c$. Based on $\Delta Q_{l,h}^{(c)}(t)$,
define the {\em backpressure of link $l$ w.r.t. flow $c$} as
\begin{align}
\Delta Q_l^{(c)}(t) \triangleq \max_{H^{\min}_{d(l)\to c}\leq h\leq
N-1} \Delta Q_{l,h}^{(c)}(t).\label{eqn:enh-DBP-BP-Shakkottai}
\end{align}
The joint traffic-splitting and shortest-path-aided DBP routing
algorithm proposed in \cite{YingShakkottai08oncombSPDBP} consists of
two parts. For traffic splitting, at time $t$, the
exogenous arrivals of flow $c$ are deposited into queue $
\left(s(c),d(c),h^*_c(t)\right)$, where $ h^*_c(t)\triangleq \arg \min_{0<h\leq N-1}
\left(Vh+Q^{(d(c))}_{s(c),h}(t)\right)$. The shortest-path-aided DBP routing is
the same as the traditional DBP routing with $\Delta Q_l^{(c)}(t) $
defined in \eqref{eqn:enh-DBP-BP-Shakkottai}. It is shown in
\cite{YingShakkottai08oncombSPDBP} that the joint traffic-splitting
and shortest-path-aided DBP routing algorithm is throughput-optimal
and solves the average path length minimization problem in
\eqref{eqn:enh-DBP-SP-prob} when $V\to \infty$. This enhanced DBP
routing algorithm achieves significant delay improvement over the
traditional DBP algorithm.

\subsection {Delay Reduction in DBP Routing by Modified Queueing Discipline}

The traditional DBP routing maintains large queue
lengths at  nodes (especially those far from the destination nodes)
so as to form gradients for data flows. This guarantees throughput
optimality while leading to poor delay
performance. In the following, we  introduce the algorithms proposed
in \cite{NeelyLMDBP11,StandfordLIFO11,NeelyLIFO11}, which try to
maintain the  gradients for data flows in DBP routing while reducing
delay for most of the packets.

In \cite{NeelyLMDBP11}, the proposed {\em fast
quadratic Lyapunov based algorithm} (FQLA) can achieve $[\mathcal O
(1/V)$, $O(\log^2(V))]$ utility-delay tradeoff, which is greatly
improved compared with $[\mathcal O (1/V),O(V)]$ utility-delay
tradeoff of the traditional DBP routing algorithm (also called quadratic Lyapunov based algorithm (QLA)
in \cite{NeelyLMDBP11}). In \cite{NeelyLMDBP11}, the authors show
that under QLA, the backlog vector ``typically" stays close to an
``attractor" and the probability of the backlog vector
deviating from the attractor is exponentially
decreasing in distance. Based on this ``exponential attraction"
result, FQLA subtracts the attractor to form a virtual backlog
process and applies the traditional DBP routing based on the virtual
backlog process with slight modification by allowing packet dropping
under certain conditions. It is shown in \cite{NeelyLMDBP11} that
the FQLA is throughput-optimal and the packet drop fraction is
$\mathcal O(1/V^{\log V})$. With the sacrifice of packet dropping,
the FQLA improves the utility-delay tradeoff.

{\em LIFO DBP routing} is first proposed in the
empirical work \cite{StandfordLIFO11} by simply replacing the FIFO
in the traditional DBP routing with the LIFO service discipline. The
authors in \cite{StandfordLIFO11} show that LIFO DBP routing
drastically improves average delay by simulations. Using the
``exponential attraction" result developed in \cite{NeelyLMDBP11},
Neely shows in \cite{NeelyLIFO11} that the LIFO DBP routing
algorithm can achieve $[\mathcal O (1/V),O(\log^2(V))]$
utility-delay tradeoff for almost all the arrival packets except
$\mathcal O(1/V^{\log V})$ fraction of the arrival packets. The
reason is as follows. The FIFO and LIFO DBP routing result in the
same queue process. By the ``exponential attraction" result in
\cite{NeelyLMDBP11}, the queue size under  DBP routing will mostly
fluctuate within the interval $[Q_{\text{Low}},Q_{\text{High}}]$,
the length of which is shown to be $O(\log^2(V))$. The queue process
deviates this region with probability exponentially decreasing in
distance. Using LIFO, most packets (except $\mathcal O(1/V^{\log
V})$ of the arrivals) enter and leave the queue when the queue
length is in $[Q_{\text{Low}},Q_{\text{High}}]$, i.e., they ``see" a
queue with average queue length about
$Q_{\text{High}}-Q_{\text{Low}}=O(\log^2(V))$. Therefore, the
average delay of these packets is greatly reduced with the penalty
that the packets of fraction $\mathcal O(1/V^{\log V})$ of the
arrivals at the front of the queue suffer from large delay and have
to be dropped.

\subsection {Delay Reduction in DBP Routing by Receiver Diversity}

Under unreliable channel conditions, the
traditional DBP routing, which makes routing decisions before each
transmission, fails to exploit multi-receiver diversity in wireless
networks. In the following, we   discuss the routing algorithms,
which  use  the receiver diversity under unreliable channel
conditions by routing packets to the successful receivers after each
transmission\cite{JavidiORCD1:2010,JavidiORCD2:2010,Neely09DIVBAR}.

The ExOR proposed in
\cite{Biswas05exor:opportunistic} is a shortest path routing
algorithm which uses expected transmission counting
metric (ETX) as the metric of link cost and chooses
the receiver with the minimum ETX after each transmission. Thus, it
can achieve better delay performance than the shortest path routing
algorithm using ETX with routing decision made before transmission
\cite{Couto03ahigh-throughput}. However, ExOR is not throughput
optimal. In \cite{JavidiORCD1:2010,JavidiORCD2:2010}, the authors
propose the {\em opportunistic routing with congestion diversity}
(ORCD) algorithm for multi-hop wireless
networks with multiple sources and a single destination. ORCD is a
shortest path routing algorithm with the queue length based congestion
measure as the path length metric, and routes the packets along the
paths with the minimum congestion after transmission. It is shown
that ORCD is throughput-optimal.

{\em Diversity backpressure routing} (DIVBAR)
algorithm proposed in \cite{Neely09DIVBAR} is a DBP routing
algorithm exploiting receiver diversity in  multi-hop wireless
networks with multiple sources and multiple destinations. In DIVBAR,
the backpressure of each node $n$ w.r.t. commodity $c$, i.e.,
$\Delta Q_n^{(c)}(t)$, is defined  as the success probability
weighted sum of  $\Delta Q_l^{(c)}(t)$ defined in
\eqref{eqn:tra-DBP-BP} over all link $l$ with $s(l)=n$. Then, the
optimal commodity of node $n$ is defined as $c^*_n(t)\triangleq \arg
\max_{c\in\mathcal C}\Delta Q_n^{(c)}(t)$. The resource allocation of DIVBAR
based on $c^*_n(t)$ is similar to that of the traditional DBP
algorithm, while packets are routed to the receiver with the largest
positive $\Delta Q_l^{(c^*_n(t))}(t)$ among all the successful
receivers after each transmission. Like traditional DBP routing,
DIVBAR is throughput-optimal.

\section{Comparisons}
\label{sec_comparison}

In this section, we   compare the three approaches
in dealing with delay sensitive resource allocation
using the uplink OFDMA system example as illustration.

\subsection {Comparison of Solution Structures and Complexity}
In general, the solution obtained by the first approach (equivalent
rate constraint) is adaptive to the CSI only. On the other hand, the
solution obtained by the second approach (Lyapunov stability drift)
and the third approach (MDP) is adaptive to both the CSI and the QSI
but the MDP approach has higher complexity. Using the uplink OFDMA system as an
example, the solution structure of the second and third approaches
are quite similar. For the Lyapunov drift approach, the solution is
obtained by the one-hop dynamic backpressure algorithm (M-LWDF) in
Example \ref{Exm:LD} with the following optimization problem
formulation:
\begin{align}
\max_{\Omega = \{ \Omega_{l,s}, \Omega_{l,p} \}}&\  \sum_{l=1}^L
Q_l\left(\sum_{m=1}^{N_F}s_{l,m}\log(1+p_{l,m}|H_{l,m}|^2)\right), \quad \forall (\mathbf H,\mathbf Q) \in \mathcal H \times \mathcal Q\label{eqn:ch6-LD-obj}\\
s.t. \ &  s_{l,m}\in\{0,1\}, \forall l\in \mathcal L, m\in \{1,2,\cdots, N_F\},\quad \sum_{l=1}^N s_{l,m} =1, \quad  \forall l\in\mathcal L \nonumber \\
&\mathbf{E}^{\Omega}\left[\sum_{m=1}^{N_F}p_{l,m}\right]\leq P_l, \quad
\forall l\in\mathcal L. \nonumber
\end{align}
For the MDP approach, the solution is obtained by solving the
Bellman equation (\ref{eqn:ch_5_bellman}) in Example
\ref{exam:app-v} with the following equivalent optimization problem
formulation:
\begin{align}
\max_{\Omega = \{ \Omega_{l,s}, \Omega_{l,p} \}}&\ \sum_{l=1}^L
\frac{\tau}{\overline{N}_l}\Delta \widetilde{V}(Q_l)
\left(\sum_{m=1}^{N_F}s_{l,m}\log(1+p_{l,m}|H_{l,m}|^2)\right),\quad
\forall (\mathbf H,\mathbf Q) \in \mathcal H \times \mathcal
Q\label{eqn:ch6_MDP-obj}\\
s.t. \ &  s_{l,m}\in\{0,1\}, \forall l\in \mathcal L, m\in \{1,2,\cdots, N_F\},\quad  \sum_{l=1}^N s_{l,m} =1, \quad  \forall l\in \mathcal L\nonumber \\
&\mathbf{E}^{\Omega}\left[\sum_{m=1}^{N_F}p_{l,m}\right]\leq P_l, \quad
\forall l\in \mathcal L. \nonumber
\end{align}
Observe that the M-LWDF problem in \eqref{eqn:ch6-LD-obj} is very
similar to the MDP problem in \eqref{eqn:ch6_MDP-obj} except that
the weight\footnote{Note that the factor
$\frac{\tau}{\overline{N}_l}$ represents the transformation from
packet flow (considered in Example \ref{exam:app-v}) to bit flow
(considered in ) Example \ref{Exm:LD}.} for the $l$-th link
(i.e., throughput of the $l$-th MS) in the later
case \eqref{eqn:ch6_MDP-obj} is given by the potential function
$\Delta \widetilde{V}(Q_l)$ whereas the weight in the former case
\eqref{eqn:ch6-LD-obj} is given by the queue state $Q_l$. The subcarrier allocation allocation in all the three
approaches will select the subcarrier with the highest metric. The
metric in the equivalent rate constraint approach is a function of the
CSI only whereas the metrics in the other two approaches are
functions of the CSI and the QSI.  The power control in all the
three approaches have a similar form of power water-filling w.r.t.
the CSI. However, the {\em water-level} of each link in the
equivalent rate constraint approach only adapts to the Lagrange
multiplier corresponding to the average delay constraint in
\eqref{eqn:ch3_delay-constraint} of each queue. In other words, the
water-levels of different links are different in general (with
different average delay requirements), while the water-level of the
same link remains constant during different
realtime system state $(\mathbf H (t), \mathbf Q(t))$ realizations.
However,  in the other two approaches, the water-level of each link
varies according to realtime system state $(\mathbf
H(t), \mathbf Q(t))$ realizations. Specifically, in the Lyapunov drift approach,
the water level of the $l$-th link is determined by the QSI
$Q_l(t)$. In the MDP approach, the water level of the
$l$-th link is determined by the QSI via the potential function
$\Delta \widetilde{V}(Q_l(t))$, which is obtained by solving the
Bellman equation  in (\ref{eqn:ch_5_bellman}) of Example
\ref{exam:app-v}. As a result, the major difference between the
second approach and the third approach is on the calculation of weight. In
the third  approach, additional processing is involved to compute
the potential functions and this contributes to additional
complexity.

\subsection{Comparison on Distributed Implementation}

In this part, we  discuss the feasibility of the distributed
implementations using different approaches. First of all, using the
first approach, the optimization problem will be transformed into a
standard convex optimization problem (such as in Example
\ref{exam:rate}). As such, traditional primal-dual decomposition
techniques may be used to explore distributed implementations. For
example, in Example \ref{exam:rate}, the subcarrier allocation
can be done by a distributed auction mechanism, and the power
allocation at each user can be calculated locally according to the
auction results and the local CSI. Readers could refer to
\cite{Palomar:07} for a survey on the decomposition method. On the
other hand, using the second approach, the one-hop Dynamic
Backpressure Algorithm (M-LWDF)  only requires the local system
state information at the transmitter and therefore, the M-LWDF
problem can be solved distributively. For example, in Example
\ref{Exm:LD}, the subcarrier allocation can be done by a
distributed auction mechanism where the auction bid is determined by
the local CSI and the local QSI, and power allocation at each user
can be calculated locally according to the auction results and the local
system information. However, in multi-hop networks, the QSI of the
neighboring nodes is required at each node, raising
additional signaling overhead on the
distributed implementation. Finally, using the third approach,
the  obstacle of the distributed implementation comes from the potential
function $V(\chi)$ and the transition kernel term
$\Pr[\chi'|\chi,\Omega(\chi)]$. In general, these terms are not
decomposable and this poses an additional challenge (compared with
the second approach) of getting a distributed solution using the MDP
approach. In Section \ref{sec_MDP_sub_appro_q}, we have illustrated
that by approximating  potential functions or Q-factors as the sum
of per-link potential functions or Q-factors, the
distributed solution can be obtained  via an auction mechanism.

\subsection{Comparison of Performance}
In this section, we   compare the performance of the three
approaches using the uplink OFDMA system example. For simplicity, we assume $N=3$, and the buffer length
$N_Q=5$ (packets). The scheduling slot duration $\tau=1$ ms. All the
users have the same average Poisson packet arrival rate $\lambda=3$
(packets/s), and exponential packet size distribution with mean
packet size $\overline{N}=5000$ (bits/packet). The total bandwidth
is assumed to be 10MHz, with $1024$ subcarriers and $5$ independent
subbands.

Fig. \ref{sec6:fig_delay_cmp} compares the average delay performance
of the three approaches under the same average power constraints. It can
be observed that the delay performance of the MDP approach is better
than those of the equivalent rate constraint approach and the Lyapunov drift
approach in the entire operating regime. Furthermore, Fig.
\ref{sec6:fig_delay_cmp} also illustrates that the performance of
the approximated-MDP approach is very close to the
brute-force MDP solution. As a result, the approximated-MDP approach
is an acceptable way to reduce the complexity and achieve a near
optimal performance. On the other hand, the equivalent rate
constraint approach (CSI-only policy) is the simplest solution but
the gap in the delay performance is small only in the very large
delay regime. The delay performance  (and the complexity) of
the Lyapunov stability drift approach is between those of the CSI-only
approach and the MDP approach.

\begin{figure}
\centering
\includegraphics[width = 10cm]{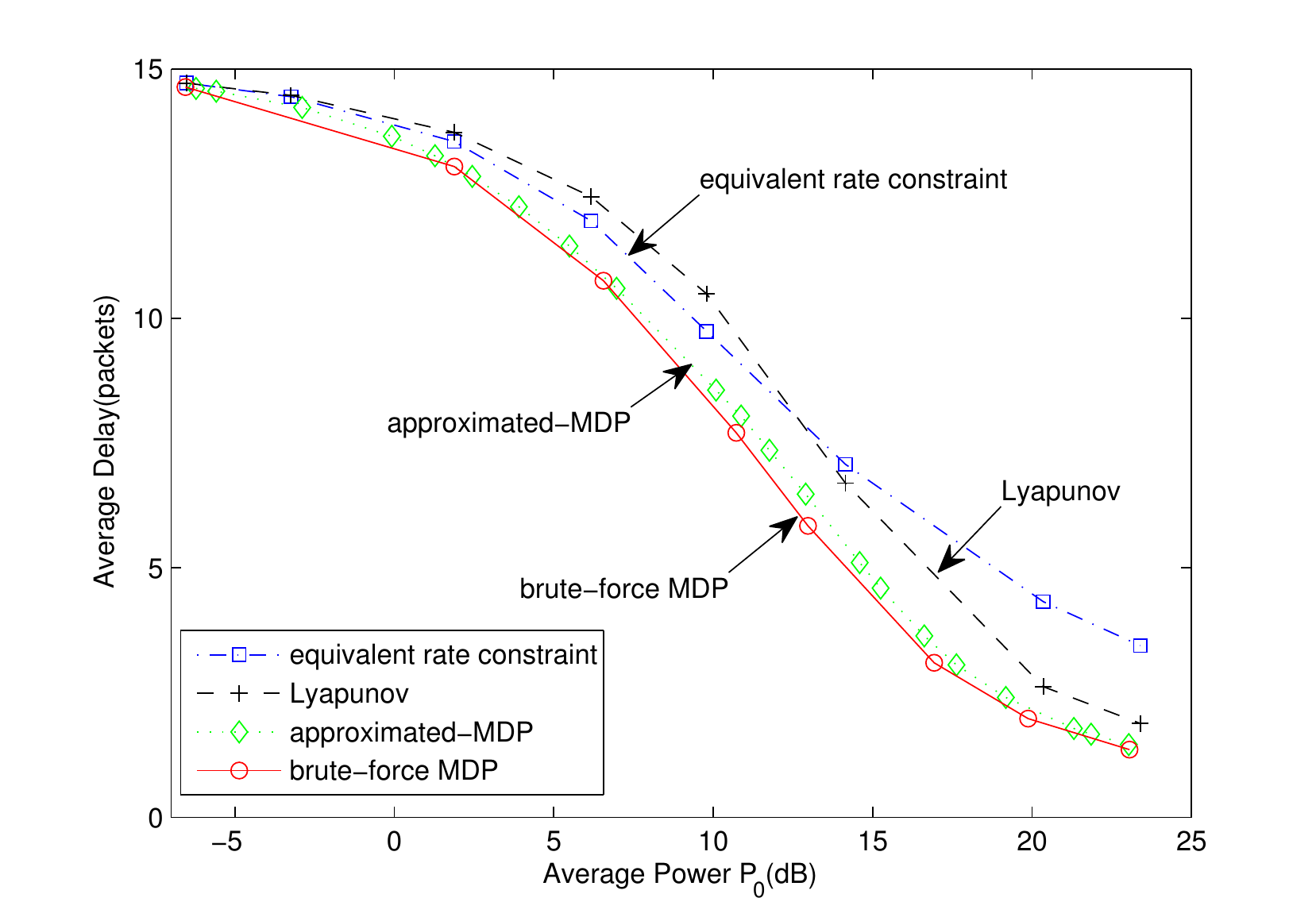}
\caption{Comparison of the delay performance of the equivalent rate constraint,  Lynapunov Stability
Drift and MDP approaches under the same average power
constraints.  The packet arrival rate is $\lambda=3$ (packets/s)
with average packet length $\overline{N}=5000$ (bits). The average
packet drop rates of all schemes are 1\%.}
\label{sec6:fig_delay_cmp}
\end{figure}

Fig. \ref{sec6:fig_num} compares the delay performance of the three
approaches  with different number of users. The average transmit SNR
for each user is   17.75dB. Similar observations about the
performance and the complexity of the three approaches can be made.

\begin{figure}
\centering
\includegraphics[width = 10cm]{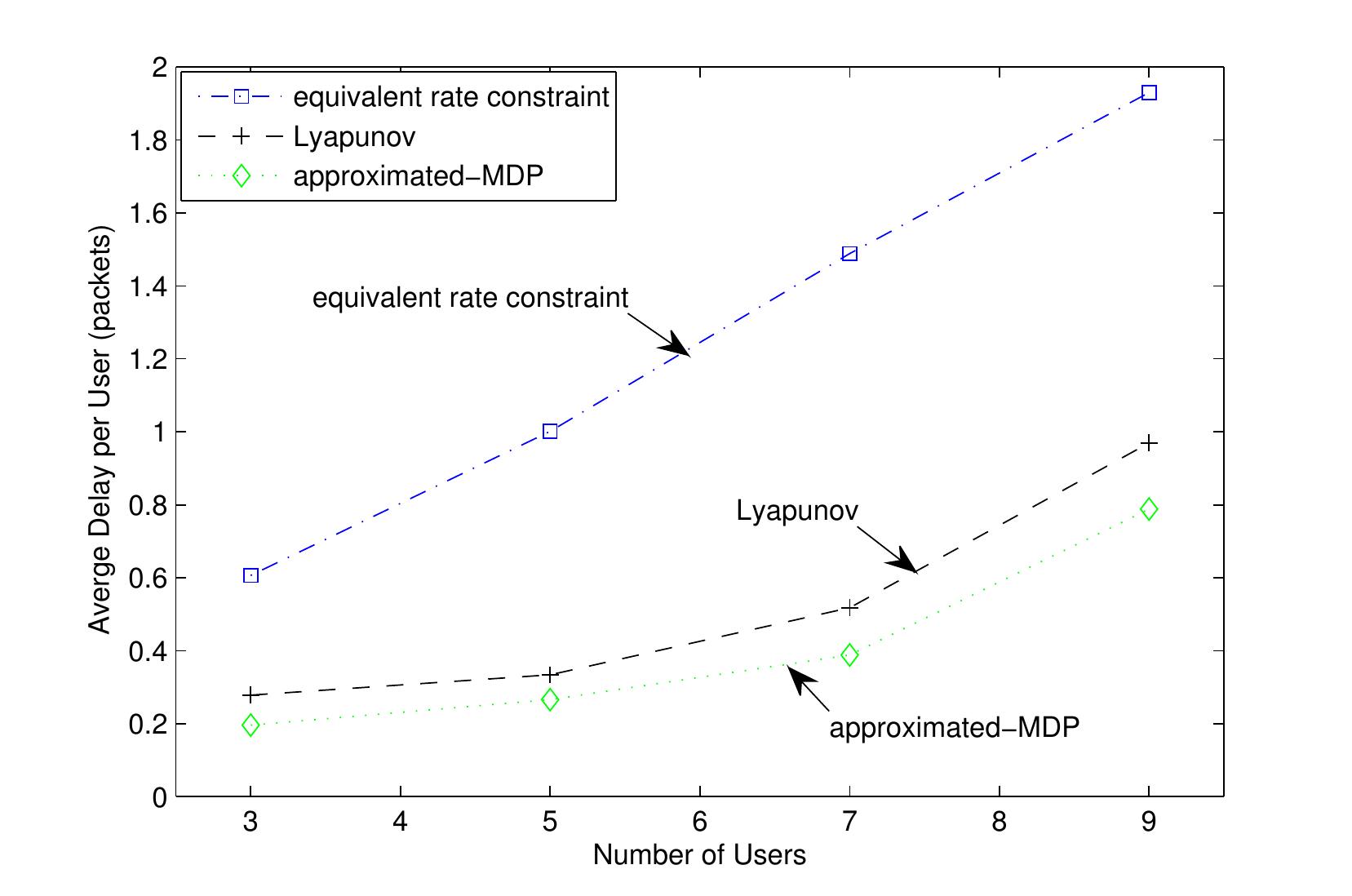}
\caption{The average delay per user versus the number of users. The average transmit power for each user is 17.75dB, and
the average Poisson packet arrival rate is $\lambda=1.5$ (packets/s) with mean packet size $\overline{N}=5000$ (bits). The
packet drop rates for all the schemes are 1\%.} \label{sec6:fig_num}
\end{figure}

Fig. \ref{sec6:fig_converge} illustrates the convergence property of
the approximate MDP approach using distributed stochastic learning. We plot the average per-link potential functions of the 3
users versus the scheduling slot index at a
transmit SNR=10dB. It can be seen that the distributed algorithm
converges quite fast. The average delay corresponding to the average
per-link potential functions at the $500$-th scheduling slot is $5.9$, which
is much smaller than those of the other baselines.

\begin{figure}
\centering
\includegraphics[width = 10cm]{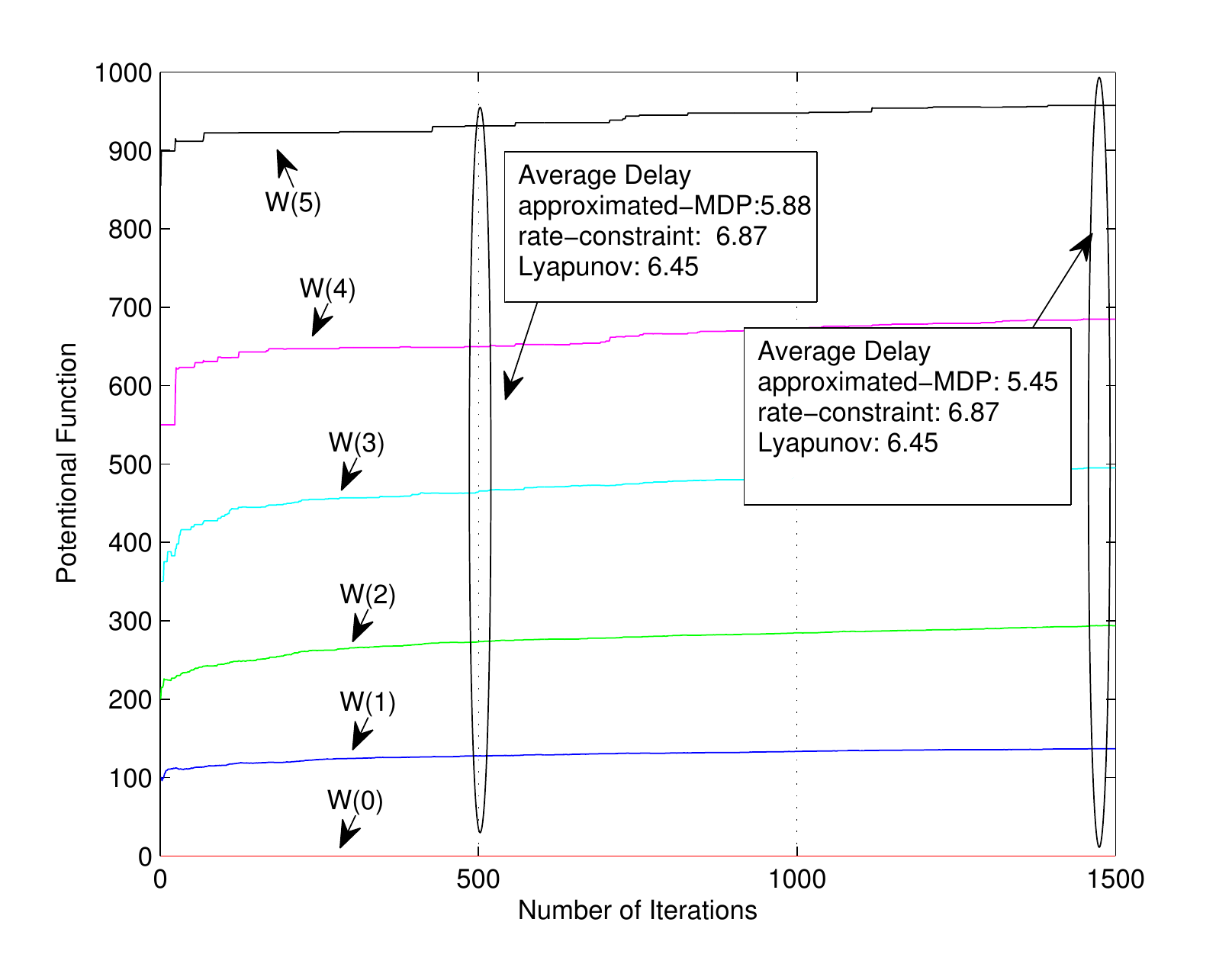}
\caption{Illustration of the convergence property for the distributed online
learning algorithm of the approximated MDP approach. Parameter
vector versus the iteration index with the average transmit SNR
14.3dB. The average Poisson packet
arrival rate is $\lambda=3$ (packets/s) with mean packet size
$\overline{N}=5000$ (bits).} \label{sec6:fig_converge}
\end{figure}

\section{Summary}
\label{sec_summary}

In this paper, we have introduced three major
approaches, namely the equivalent rate constraint approach, the
Lynapnov stability drift approach as well as the MDP approach, to
deal with delay-aware resource allocation for wireless networks. For
the MDP approach, we use the
approximated MDP and stochastic learning to solve the curse of dimensionality and facilitate
distributed online implementation. Moreover, we also elaborate
on how to use these approaches in an uplink OFDMA
system. It is shown by simulations that the equivalent rate
constraint approach performs better than the
Lynapnov stability drift approach in the large delay regime and worse
in the small delay regime, and the MDP approach has much better delay performance than the other
two schemes in all regimes.

\section*{Acknowledgment}

The authors would like to thank the anonymous reviewers, Prof. Edmund M. Yeh, Eddy Chiu, An Liu, Fan Zhang and Farrah Mckay
for their valuable comments.

\section*{Appendix A: Proof of Lemma \ref{lem:conv}}

Since each representative state is updated comparably
often\footnote{Please refer to \cite{Abounadi:98} for the definition
of ``comparably often".} in the asynchronous learning algorithm,
quoting the conclusion in \cite{Borkar:98}, the convergence property
of the asynchronous update and synchronous update is the same.
Therefore, we consider the convergence of the
related synchronous version for simplicity in this proof.

Let $c\in R$ be a constant, we have $(T c \mathbf{M}
\widetilde{\mathbf{V}}_t) (I) = c (T \mathbf{M}
\widetilde{\mathbf{V}}_t) (I)$. Similar to \cite{Borkar:00}, it is easy to see that the parameter vector
$\{\widetilde{\mathbf{V}}_t\}$ is bounded almost surely during the
iterations of the algorithm. In the following, we
first introduce and prove the following lemma on the convergence of
learning noise.

\begin{Lem} Define
\begin{equation}
\mathbf{q}_t = \mathbf{M}^{-1} \bigg[ \mathbf{g}(\Omega_t) + \mathbf{F}(\Omega_t)\mathbf{M} \widetilde{\mathbf{V}}_t -
\mathbf{M} \widetilde{\mathbf{V}}_t - (T \mathbf{M} \widetilde{\mathbf{V}}_t) (I) \mathbf{e}\bigg], \nonumber
\end{equation}
when the number of iterations $t \geq j \rightarrow \infty$, the
update procedure can be written as follows with probability 1:
\begin{equation}
\widetilde{\mathbf{V}}_{t+1} = \widetilde{\mathbf{V}}_{j} +
\sum_{i=j}^t \epsilon_t \mathbf{q}_i \nonumber.
\end{equation}\label{lem:noise}
\end{Lem}

\begin{proof}
The update of parameter vector can be written in the following vector form:
\begin{equation}
\widetilde{\mathbf{V}}_{t+1} = \widetilde{\mathbf{V}}_{t} +
\epsilon_t \mathbf{M}^{-1} \bigg[ \mathbf{g}(\Omega_t) + \mathbf{J}_t
\mathbf{M} \widetilde{\mathbf{V}}_t - \mathbf{M}
\widetilde{\mathbf{V}}_t -
\bigg(g(I,\Omega_t)+(\mathbf{M}\widetilde{\mathbf{V}}_{t})(I^+)\bigg) \mathbf{e} \bigg], \nonumber
\end{equation}
where the matrix $\mathbf{J}_t$ (with exactly one element of 1 in
each row) denotes the realtime observed state transition from the
$t$-th frame to the $t+1$-th frame, and $ I^+ $ denotes the observed
next state if the current state is $ I $. Define
\begin{equation}
\mathbf{Y}_t = \mathbf{M}^{-1} \bigg[ \mathbf{g}(\Omega_t) +
\mathbf{F}(\Omega_t) \mathbf{M} \widetilde{\mathbf{V}}_t -
\mathbf{M} \widetilde{\mathbf{V}}_t - (T \mathbf{M}
\widetilde{\mathbf{V}}_t)(I) \mathbf{e} \bigg], \nonumber
\end{equation}
and $\delta \mathbf{Z}_t = \mathbf{Y}_t - \mathbf{q}_t$ and
$\mathbf{Z}_t = \sum\limits_{i=j}^t \epsilon_i \delta \mathbf{Z}_i$.
The online potential estimation can be rewritten as
\begin{eqnarray}
\widetilde{\mathbf{V}}_{t+1} &=& \widetilde{\mathbf{V}}_{t} + \epsilon_t \mathbf{Y}_t \nonumber\\
&=& \widetilde{\mathbf{V}}_{t} + \epsilon_t \mathbf{q}_t - \epsilon_t \delta \mathbf{Z}_t \nonumber \\
&=& \widetilde{\mathbf{V}}_{t} + \sum_{i=j}^t \epsilon_i
\mathbf{q}_i - \mathbf{Z}_t \label{eqn:it-V}.
\end{eqnarray}
 Our proof of Lemma \ref{lem:noise} can be divided into the following steps:
\begin{itemize}
\item[1.] Letting $\mathcal{F}_t=\sigma(\widetilde{\mathbf{V}}_m,m\leq t)$, it is easy to see that
$\mathbf{E}[\delta \mathbf{Z}_t|\mathcal{F}_{t-1}]=0$. Thus,
$\{\delta \mathbf{Z}_t|\forall t\}$ is a Martingale difference
sequence and $\{\mathbf{Z}_t|\forall t\}$ is a Martingale sequence.
Moreover, $\mathbf{Y}_t$ is an unbiased estimation
of $\mathbf{q}_t$ and the estimation noise is uncorrelated.

\item[2.] According to the uncorrelated estimation error from Step 1, we have
\begin{eqnarray}
\mathbf{E}\bigg[|\mathbf{Z}_t|^2 \bigg| \mathcal{F}_{j-1} \bigg] &=& \mathbf{E}\bigg[|\sum\limits_{i=j}^t \epsilon_i \delta \mathbf{Z}_i|^2 \bigg|\mathcal{F}_{j-1} \bigg] \nonumber\\
&=& \sum\limits_{i=j}^t \mathbf{E}\bigg[|\epsilon_i \delta \mathbf{Z}_i|^2\bigg|\mathcal{F}_{j-1}\bigg] \nonumber\\
&=& \widetilde{\mathbf{Z}} \sum\limits_{i=j}^t (\epsilon_i)^2
\rightarrow 0 \quad \mbox{when} \quad j\rightarrow \infty, \nonumber
\end{eqnarray}
where $\widetilde{\mathbf{Z}}\geq \max\limits_{j\leq i \leq
t}\mathbf{E}\bigg[|\delta \mathbf{Z}_i|^2
\bigg|\mathcal{F}_{j-1}\bigg]$ is a bounded constant vector and the
convergence of $\widetilde{\mathbf{Z}} \sum\limits_{i=j}^t
(\epsilon_i)^2$ is from the definition of sequence $\{\epsilon_i\}$.

\item[3.] From Step 1, $\{\mathbf{Z}_t|\forall t\}$ is a Martingale sequence. Hence, according to the inequality of Martingale sequence, we have
\begin{equation}
\Pr\bigg[\sup\limits_{j \leq i \leq t } |\mathbf{Z}_i| \geq \lambda
\bigg| \mathcal{F}_{j-1}\bigg] \leq
\frac{\mathbf{E}\bigg[|\mathbf{Z}_t|^2\bigg|
\mathcal{F}_{j-1}\bigg]}{\lambda^2} \quad \forall \lambda>0.
\nonumber
\end{equation}
From the conclusion of Step 2, we have
\begin{equation}
\lim_{j\rightarrow \infty} \Pr\bigg[\sup\limits_{j \leq i \leq t }
|\mathbf{Z}_i| \geq \lambda \bigg| \mathcal{F}_{j-1} \bigg] =0 \quad
\forall \lambda>0. \nonumber
\end{equation}
Hence, from (\ref{eqn:it-V}) we almost surely have
$\widetilde{\mathbf{V}}_{t+1} = \widetilde{\mathbf{V}}_{j} +
\sum_{i=j}^t \epsilon_i \mathbf{q}_i $ when $j\rightarrow \infty$.
\end{itemize}
\end{proof}

Moreover, the following lemma is about the limit of sequence $\{\mathbf{q}_t\}$.

\begin{Lem} Suppose the following two inequalities are true for $l=a,a+1,...,a+b$
\begin{eqnarray}
\mathbf{g}(\Omega_l) + \mathbf{F}(\Omega_l) \mathbf{M} \widetilde{\mathbf{V}}_l &\leq &
\mathbf{g}(\Omega_{l-1}) + \mathbf{F}(\Omega_{l-1}) \mathbf{M} \widetilde{\mathbf{V}}_l \label{eqn:l}\\
\mathbf{g}(\Omega_{l-1}) + \mathbf{F}(\Omega_{l-1}) \mathbf{M} \widetilde{\mathbf{V}}_{l-1} &\leq &
\mathbf{g}(\Omega_{l}) + \mathbf{F}(\Omega_{l}) \mathbf{M} \widetilde{\mathbf{V}}_{l-1} \label{eqn:l-1},
\end{eqnarray}
then we have
\begin{equation}
|q^i_{a+b}|\leq C_1 \prod_{i=0}^{\lfloor \frac{b}{\beta} \rfloor -
1}(1-\delta_{a+i \beta}) \quad \forall i, \label{eqn:con_cov}
\end{equation}
where $q^i_{a+b}$ denotes the $i$th element of the vector $\mathbf{q}_{a+b}$, $C_1$ is some constant.
\end{Lem}

\begin{proof}
From (\ref{eqn:l}) and (\ref{eqn:l-1}), we have
\begin{equation}
\mathbf{q}_l =  \mathbf{M}^{-1} \bigg[ \mathbf{g}(\Omega_l) + \mathbf{F}(\Omega_l)\mathbf{M} \widetilde{\mathbf{V}}_l -
\mathbf{M} \widetilde{\mathbf{V}}_l - w_l \mathbf{e}\bigg] \leq \mathbf{M}^{-1} \bigg[ \mathbf{g}(\Omega_{l-1}) +
\mathbf{F}(\Omega_{l-1}) \mathbf{M} \widetilde{\mathbf{V}}_l - \mathbf{M} \widetilde{\mathbf{V}}_l - w_l
\mathbf{e}\bigg] \nonumber
\end{equation}
\begin{eqnarray}
\mathbf{q}_{l-1} &=&  \mathbf{M}^{-1} \bigg[ \mathbf{g}(\Omega_{l-1}) + \mathbf{F}(\Omega_{l-1}) \mathbf{M}
\widetilde{\mathbf{V}}_{l-1} - \mathbf{M} \widetilde{\mathbf{V}}_{l-1} - w_{l-1} \mathbf{e}\bigg] \nonumber \\
&\leq&  \mathbf{M}^{-1} \bigg[ \mathbf{g}(\Omega_{l}) + \mathbf{F}(\Omega_{l}) \mathbf{M} \widetilde{\mathbf{V}}_{l-1}
- \mathbf{M} \widetilde{\mathbf{V}}_{l-1} - w_{l-1} \mathbf{e}\bigg] \nonumber
\end{eqnarray}
where $w_l = (T \mathbf{M} \widetilde{\mathbf{V}}_l)(I)$. According to Lemma \ref{lem:noise}, we have
\begin{equation}
\widetilde{\mathbf{V}}_l = \widetilde{\mathbf{V}}_{l-1} +
\epsilon_{l-1} \mathbf{q}_{l-1} \Rightarrow \widetilde{\mathbf{V}}_l
= \widetilde{\mathbf{V}}_{l-1} + \epsilon_{l-1}
\mathbf{q}_{l-1}\nonumber,
\end{equation}
therefore,
\begin{eqnarray}
\mathbf{q}_l &\leq& \bigg[(1-\epsilon_{l-1})\mathbf{I} +
\mathbf{M}^{-1} \mathbf{F}(\Omega^{l-1}) \mathbf{M} \epsilon_{l-1}
\bigg] \mathbf{q}_{l-1} + w_{l-1} \mathbf{e} - w_l \mathbf{e}
  = \mathbf{B}_{l-1} \mathbf{q}_{l-1} + w_{l-1} \mathbf{e} - w_l \mathbf{e} \nonumber\\
\mathbf{q}_l &\geq & \bigg[(1-\epsilon_{l-1})\mathbf{I} +
\mathbf{M}^{-1} \mathbf{F}(\Omega^{l}) \mathbf{M} \epsilon_{l-1}
\bigg] \mathbf{q}_{l-1} + w_{l-1} \mathbf{e} - w_l \mathbf{e} =
\mathbf{A}_{l-1} \mathbf{q}_{l-1} + w_{l-1} \mathbf{e} - w_l
\mathbf{e}.\nonumber
\end{eqnarray}
Notice that
\begin{equation}
\mathbf{A}_{l-1}\mathbf{e} = (1-\epsilon_{l-1})\mathbf{I} \mathbf{e}
+  \mathbf{M}^{-1} \mathbf{F}(\Omega^{l}) \mathbf{M} \epsilon_{l-1}
\mathbf{e} = (1-\epsilon_{l-1})\mathbf{e} + L\epsilon_{l-1}
\mathbf{e} \nonumber
\end{equation}
\begin{equation}
\mathbf{B}_{l-1}\mathbf{e} = (1-\epsilon_{l-1})\mathbf{I} \mathbf{e}
+  \mathbf{M}^{-1} \mathbf{F}(\Omega^{l-1}) \mathbf{M}
\epsilon_{l-1} \mathbf{e} = (1-\epsilon_{l-1})\mathbf{e} +
L\epsilon_{l-1} \mathbf{e} \nonumber,
\end{equation}
where $ L $ is the total number of links in the network. Notice that $\mathbf{A}_{l-1}\mathbf{e}=\mathbf{B}_{l-1}\mathbf{e}$,
we have
\begin{eqnarray}
&&\mathbf{A}_{l-1}...\mathbf{A}_{l-\beta} \mathbf{q}_{l-\beta} - C_1
\mathbf{e} \leq \mathbf{q}_l \leq
\mathbf{B}_{l-1}...\mathbf{B}_{l-\beta} \mathbf{q}_{l-\beta} - C_1 \mathbf{e} \nonumber\\
&\Rightarrow& (1-\delta_l) [\min \mathbf{q}_{l-\beta}]\leq
\mathbf{q}_l + C_1 \mathbf{e} \leq (1-\delta_l) [\max
\mathbf{q}_{l-\beta}]
\nonumber\\
&\Rightarrow& \begin{cases}
    \max \mathbf{q}_l + C_1 \leq (1-\delta_l) \max \mathbf{q}_{l-\beta}  \nonumber\\
    \min \mathbf{q}_l + C_1 \geq (1-\delta_l) \min \mathbf{q}_{l-\beta}\end{cases} \nonumber\\
&\Rightarrow& \max \mathbf{q}_l - \min \mathbf{q}_l \leq
(1-\delta_l) \bigg[ \max \mathbf{q}_{l-\beta} - \min
\mathbf{q}_{l-\beta}\bigg]
\nonumber\\
&\Rightarrow & |q^l_i| \leq \max \mathbf{q}_l - \min \mathbf{q}_l
\leq C_2 (1-\delta_l) \quad \forall i \nonumber,
\end{eqnarray}
where the first step is due to conditions on matrix sequence
$\{\mathbf{A}_l\}$ and $\{\mathbf{B}_l\}$, $\max \mathbf{q}_l$ and
$\min \mathbf{q}_l$ denote the maximum and minimum elements in
$\mathbf{q}_l$ respectively, $C_1$ and $C_2$ are all constants, the
first inequality of the last step is because $\min \mathbf{q}_l \leq
0$. Hence, the conclusion is straightforward.
\end{proof}

Therefore, the proof of Lemma \ref{lem:conv} can be divided into the following steps:

\begin{itemize}
\item[1.] From the property of sequence $\{\epsilon_t\}$, we have $\prod_{i=0}^{\lfloor \frac{t}{\beta} \rfloor - 1}(1-\epsilon_{i \beta}) \rightarrow 0$ ($t \rightarrow
\infty$).
\item[2.] According to the first step, note that $\delta_t=\mathcal{O}(\epsilon_t)$, from (\ref{eqn:con_cov}), we have $\mathbf{q}_t \rightarrow 0$ ($t \rightarrow \infty$).
\item[3.] Therefore, the update on $\{\widetilde{\mathbf{V}}_l\}$ will converge to $\widetilde{\mathbf{V}}_{\infty}$, which satisfies the following fixed-point equation
\begin{equation}
\theta \mathbf{e} + \widetilde{\mathbf{V}}_{\infty} = \mathbf{M}^{-1} \mathbf{T} ( \mathbf{M}
\widetilde{\mathbf{V}}_{\infty} ). \nonumber
\end{equation}
\end{itemize}
This completes the proof.

%

 \begin{biographynophoto}{Ying~Cui (S'08)}
received B.Eng degree (first class honor) in Electronic and
Information Engineering, Xi'an Jiaotong University, China in 2007.
She is currently a Ph.D candidate in the Department of ECE, the Hong
Kong University of Science and Technology (HKUST). Her current
research interests include cooperative and cognitive communications,
delay-sensitive cross-layer scheduling as well as stochastic
approximation and Markov Decision Process.
\end{biographynophoto}

\begin{biographynophoto}{Vincent~K.~N.~Lau (SM'01)}
obtained B.Eng (Distinction 1st Hons) from the University of Hong
Kong in 1992 and Ph.D. from Cambridge University in 1997. He was
with PCCW as system engineer from 1992-1995 and Bell Labs - Lucent
Technologies as member of technical staff from 1997-2003. He then
joined the Department of ECE, HKUST and is currently a Professor. His research interests include the robust and delay-sensitive
cross-layer scheduling, cooperative and cognitive communications as
well as stochastic optimization.
\end{biographynophoto}

\begin{biographynophoto}{Rui~Wang (S'04-M'09)}
received B.Eng degree (first class honor) in Computer Science from
the University of Science and Technology of China in 2004 and Ph.D
degree in the Department of ECE from HKUST in 2008. He is currently
a senior research engineer in Huawei Technologies Co. Ltd.
\end{biographynophoto}

\begin{biographynophoto}{Huang~Huang (S'08)}
received the B.Eng. and M.Eng. (Gold medal) from the Harbin
Institute of Technology(HIT) in 2005 and 2007, respectively, and Ph.D from
The Hong Kong University of Science and Technology (HKUST) in 2011. He is
currently a research engineer in Huawei Technologies Co. Ltd.
\end{biographynophoto}

\begin{biographynophoto}{Shunqing~Zhang (S'05-M'09)}
obtained B.Eng from Fudan University and Ph.D. from Hong Kong University
of Science and Technology (HKUST) in 2005 and 2009, respectively. He joined Huawei
Technologies Co., Ltd in 2009, where he is now the system engineer of Green Radio Excellence
in Architecture and Technology (GREAT) team. His current research interests include the
energy consumption modeling of the wireless system, the energy efficient wireless
transmissions as well as the energy efficient network architecture and protocol design.
\end{biographynophoto}

\end{document}